\documentclass[journal,twoside,web]{ieeecolor}
\usepackage{generic}
\usepackage{cite}
\usepackage{amsmath,amssymb,amsfonts}
\usepackage{algorithmic}
\usepackage{graphicx}
\usepackage{algorithm,algorithmic}
\usepackage{hyperref}
\hypersetup{hidelinks=true}
\usepackage{textcomp}
\usepackage{theoremref}
\usepackage{makecell}
\def\BibTeX{{\rm B\kern-.05em{\sc i\kern-.025em b}\kern-.08em
    T\kern-.1667em\lower.7ex\hbox{E}\kern-.125emX}}
\newcommand{\x}{\boldsymbol{x}}
\newcommand{\z}{\boldsymbol{z}}
\newcommand{\bmu}{\boldsymbol{\mu}}
\newcommand{\w}{\boldsymbol{w}}

\newcommand{\s}{\boldsymbol{s}}
\newcommand{\dd}{\boldsymbol{d}}
\DeclareMathOperator*{\argmax}{arg\,max}
\DeclareMathOperator*{\argmin}{arg\,min}

\newtheorem{remark}{Remark}
\newtheorem{definition}{Definition}
\newtheorem{lemma}{Lemma}
\newtheorem{theorem}{Theorem}

\begin{document}
\title{Performance Guarantees for Data-Driven Sequential Decision-Making}
\author{Bowen Li, \IEEEmembership{Graduate Student Member, IEEE}, Edwin K. P. Chong, \IEEEmembership{Fellow, IEEE}, and Ali Pezeshki, \IEEEmembership{Senior Member, IEEE}
\thanks{Manuscript submitted on March 31st, 2026.}
\thanks{This work is supported in part by the AFOSR under award FA8750-20-2-0504 and by the NSF under award CNS-2229469.}
\thanks{The authors are with the Department of Electrical and Computer Engineering, Colorado State University, Fort Collins, CO 80523, USA. (email: bowen.li@colostate.edu; edwin.chong@colostate.edu; ali.pezeshki@colostate.edu).
}
}

\maketitle

\begin{abstract}
    The solutions to many sequential decision-making problems are characterized by dynamic programming and Bellman's principle of optimality. However, due to the inherent complexity of solving Bellman's equation exactly, there has been significant interest in developing various approximate dynamic programming (ADP) schemes to obtain near-optimal solutions. A fundamental question that arises is: how close are the objective values produced by ADP schemes relative to the true optimal objective values? In this paper, we develop a general framework that provides performance guarantees for ADP schemes in the form of ratio bounds. Specifically, we show that the objective value under an ADP scheme is at least a computable fraction of the optimal value. We further demonstrate the applicability of our theoretical framework through several applications: data-driven robot path planning, pendulum stabilization, and multi-agent sensor coverage.
\end{abstract}


%%%%%%%%%%%%%%%%%%%%%%%%%%%%%%%%%%%%%%%%%%%%%%%%%%%%%%%%%%%%%%
%%%%%%%%%%%%%%%%%%%%%%%%%%%%%%%%%%%%%%%%%%%%%%%%%%%%%%%%%%%%%%
\section{Introduction}
\label{sec:introduction}

In sequential decision-making and optimal control problems, the goal is typically to select a sequence of actions over a finite horizon to optimize a given objective function. Accounting for the stochastic transitions between states and the accumulation of reward (or cost) obtained at each step, such problems are commonly formulated as \textit{stochastic optimal control} problems. This formulation is often developed within the framework of \textit{Markov Decision Process (MDP)} or \textit{Partially Observable Markov Decision Process (POMDP)}, depending on whether the state of the system is fully observable \cite{chong2009partially}. The classical approach to solving such problems relies on \textit{dynamic programming} through \textit{Bellman's principle of optimality} \cite{bertsekas2012dynamic}. However, this approach suffers from the curse of dimensionality, as obtaining the exact optimal solution becomes computationally intractable with the growing size of state and action spaces, as well as longer decision horizons. To address this issue, various \textit{approximate dynamic programming (ADP)} schemes have been developed to obtain near-optimal solutions. The central idea of ADP is to replace the expected-value-to-go (EVTG) term in Bellman's equation, which is often practically impossible to compute in high-dimensional problems, with some computationally tractable approximations. Some commonly used ADP schemes are model predictive control (MPC) \cite{grune2008infinite}, traditional reinforcement learning \cite{watkins1989learning, suttonRLbook}, deep reinforcement learning \cite{silver2017mastering}, hindsight optimization \cite{chong2000framework, wu2002burst}, policy rollout \cite{bertsekas1999rollout}, etc. These techniques have been applied across a variety of domains, such as adaptive sensing, robotics, and transportation systems \cite{chong2009partially}.

Despite the widespread use of various ADP schemes, quantifying the degree of suboptimality, that is, how different the solutions produced by ADP schemes are from the optimal solutions, remains a challenging problem. Bertsekas \cite{bertsekas2012dynamic, bertsekas2005dynamic} studies bounds on the difference between the objective values achieved by ADP schemes and those of the optimal solution for both finite-horizon and infinite-horizon settings. Gr\"{u}ne and Rantzer \cite{grune2008infinite} provide bounds on the performance gap between a finite-horizon model predictive controller and the infinite-horizon optimal controller. Liu \emph{et al.}~\cite{liuadplcss}, inspired by the seminal results in submodular optimization, develop bounds for the ratio of the objective functions of ADP schemes to those of optimal solutions. However, their bound is computationally intractable. 

In this paper, we develop a general framework to provide computable performance guarantees of ADP schemes relative to optimal policies. More specifically, we are interested in quantifying the following ratio bound: 
$$
\frac{\hat{V}}{V^{*}} \geq \beta, 
$$
where $\hat{V}$ and $V^{*}$ are the expected cumulative reward under a given ADP scheme and the optimal policy, respectively. 

The \textbf{main contributions} of this paper are as follows: 

\begin{enumerate}
    \item To the best of our knowledge, the bounds developed here are 
    the first computable ratio bounds for ADP schemes that is applicable to a broad class of problems.
    \item We establish a connection to the ratio bounds for greedy algorithms in submodular optimization, and show that our ADP bounding framework generalizes those classical results. This highlights a structural link between greedy decision-making and ADP.  
    \item We demonstrate our bounding framework on practical applications, including robot path planning, pendulum stabilization, and multi-agent sensor coverage. 
\end{enumerate}

We note that our bounding framework is different from that of Liu \emph{et al.}~\cite{liuadplcss}, who also derive ratio bounds for ADP schemes, in two important ways: (1) Determining the value of the bound in Liu \emph{et al.} is computationally intractable, whereas our bounds can be easily computed. (2) Our derivation technique for the bound is different from that of Liu \emph{et al.} 

The remainder of the paper is organized as follows. Section~\ref{sec:preliminaries} introduces the notation and background material for stochastic optimal control, along with relevant prior results on the performance guarantees of ADP schemes. Section~\ref{sec:main_results} presents our main contributions and illustrates how our bounding framework generalizes some results in submodular optimization. Section~\ref{sec:applications} demonstrates the applicability of our framework to robot path planning and multi-agent sensor coverage problems. Section~\ref{sec:conclusion} presents our concluding remarks and potential directions for future work.

%%%%%%%%%%%%%%%%%%%%%%%%%%%%%%%%%%%%%%%%%%%%%%%%%%%%%%%%%%%%%%
%%%%%%%%%%%%%%%%%%%%%%%%%%%%%%%%%%%%%%%%%%%%%%%%%%%%%%%%%%%%%%
\section{Preliminaries}
\label{sec:preliminaries}
In this section, we discuss a general framework for finite-horizon stochastic optimal control and explain the method of approximate dynamic programming (ADP).

%%%%%%%%%%%%%%%%%%%%%%%%%%%%%%%%%%%%%%%%%%%%%%%%%%%%%%%%%%%%%%
\subsection{Stochastic Optimal Control}
Let $\mathcal{X}$ denote the set of states, $\mathcal{U}$ the set of control actions, and $H$ the length of decision horizon. Denote the state and control action at time step $k \; (k = 0,\ldots, H-1)$ by $\x_{k} \in \mathcal{X}$ and $\bmu_{k} \in \mathcal{U}$, respectively, with the initial state $\x_{0}$ given. Considering the possibility that the set of feasible control actions may vary at each time step, we define a feasibility mapping $\mathcal{U}_{k}: \mathcal{X} \xrightarrow{} 2^{\mathcal{U}}$, where $2^{\mathcal{U}}$ is the power set of $\mathcal{U}$. Then at time step $k$, the implemented action should satisfy $\bmu_{k} \in \mathcal{U}_{k}(\x_{k}) \subset \mathcal{U}$, i.e., the action $\bmu_{k}$ should belong to the set of feasible actions associated with the current state $\x_{k}$.

Next, we denote the state transition function at time step $k$ as $h_{k}: \mathcal{X} \times \mathcal{U} \times \mathcal{W} \xrightarrow{} \mathcal{X}$ with the transition given by $\x_{k+1} = h_{k}(\x_{k},\bmu_{k},\w_{k})$, where $\w_{k} \in \mathcal{W}$ are i.i.d. random variables that introduce stochasticity into the state transitions. The decision rule at each step is defined by the policy mapping $\pi_{k}: \mathcal{X} \xrightarrow{} \mathcal{U}$ such that $\bmu_{k} = \pi_{k}(\x_{k})$, where $\bmu_{k}$ is a deterministic action. Finally, the reward function at each step is $r_{k}:\mathcal{X} \times \mathcal{U} \xrightarrow{} \mathbb{R}_{+}$ with $r_{k}(\x_{k},\bmu_{k})$ representing a non-negative scalar reward after taking action $\bmu_{k}$ in state $\x_{k}$. The terminal state of the system is $\x_{H}$. Without further actions taken, the reward obtained at the terminal state $\x_{H}$ is $r_{H}(\x_{H})$. 

Based on the above notation, the \textit{stochastic optimal control problem} is formulated as: 
\begin{equation}
\label{mdp_form}
    \begin{aligned}
        \max_{(\pi_{0},\ldots,\pi_{H-1})} & \sum_{k=0}^{H-1}\mathbb{E}\left[ r_{k}\left( \x_{k},\pi_{k}(\x_{k}) \right) \right]  + \mathbb{E}\left[ r_{H}(\x_{H}) \right]\\
        \text{subject to} \quad & \x_{k+1} = h_{k}(\x_{k},\pi_{k}(\x_{k}),\w_{k}), \\
        & \pi_{k}(\x_{k}) \in \mathcal{U}_{k}(\x_{k}), \\
        & k = 0,\ldots,H-1.
    \end{aligned} 
\end{equation}

Problem \eqref{mdp_form} is also referred to as a \textit{Markov Decision Process (MDP)} in the literature, which arises in a wide range of applications including sensor management in remote sensing \cite{chong2009partially}, UAV control \cite{ragi2013uav}, and games such as Go \cite{silver2017mastering}. The policy mapping $\pi_{k}$ in an MDP is also called a \textit{state-feedback control law}, as the action $\bmu_{k}$ taken at step $k$ depends on state $\x_{k}$. So far, we have assumed the policy mapping $\pi_{k}$ to be deterministic. More generally, the framework can be extended to randomized policies, where the feedback-control action is a random variable conditioned on the state. In other words, we take $\pi_k(\x_k)$ to be a $\mathcal{U}$-valued random variable with a distribution that depends on $\x_k$.
%However, there always exists an optimal policy at each time step that is deterministic when the reward function $r_{k}(\x_{k},\bmu_{k})$ is deterministic \cite{puterman2014markov}. Moreover, we always need to select a particular control action at each time step. Therefore, we primarily focus on deterministic policies in this work. 

% The action $\bmu_{k}$ is also a random variable, which is a measurable function of $\x_{k}$ through the mapping $\pi_{k}$.

% mapping from sample space $\mathcal{X}$ to the measurable space $(\mathcal{U},2^{\mathcal{U}})$. As a result, the policy $\pi_{k}$ is often called a \textit{randomized policy}, because $\bmu_{k} = \pi_{k}(\x_{k})$ can end up with different actions, each with associated probabilities. 

A more general setting arises when the state $\x_{k}$ is only partially observable; that is, all we can observe at time step $k$ is $\boldsymbol{o}_{k}$ generated by the observation law $\boldsymbol{o}_{k} = O(\x_{k},\pi_{k-1}(\x_{k-1}),\boldsymbol{\xi}_{k})$, where $\boldsymbol{\xi}_{k}$ is the 
$k^{\text{th}}$ element of an i.i.d. sequence. This type of problem is known as a \textit{Partially Observable Markov Decision Process (POMDP)} and can be reformulated and solved as a belief-state MDP \cite{bertsekas2012dynamic,chong2009partially}. Our main theoretical results are only developed under the framework of MDP. However, it can be extended to POMDPs.

% \textcolor{blue}{POMDP or deterministic? } A degenerate setting of problem \eqref{mdp_form} is the case where the state transitions are deterministic without the random disturbances terms and we optimize 

Before proceeding to the subsequent sections, we introduce several notation and definitions that are essential for both existing results and our main contributions. 

%\begin{definition}
\textbf{Optimal Sequence of Policies and States:} We call an optimal solution to problem \eqref{mdp_form} an \textit{optimal sequence of policies} and denote it by $(\pi_{0}^{*},\ldots,\pi_{H-1}^{*})$. When this policy sequence is applied, the corresponding sequence of visited states is denoted by $(\x_{0}^{*},\ldots,\x_{H-1}^{*},\x_{H}^{*})$, which we call the \textit{optimal state path}. Accordingly, the \textit{optimal sequence of actions} is given by $(\bmu_{0}^{*},\ldots,\bmu_{H-1}^{*})$, where each action satisfies $\bmu_{k}^{*} = \pi_{k}^{*}(\x_{k}^{*}) \in \mathcal{U}(\x_{k}^{*})$.
%\end{definition}

Due to the randomness in the state transitions at each step, the optimal state path is a stochastic sequence. Each realization of the path $(\x_{0}^{*},\ldots,\x_{H-1}^{*},\x_{H}^{*})$ under the optimal policy sequence $(\pi_{0}^{*},\ldots,\pi_{H-1}^{*})$ may be different.  Similarly, the optimal sequence of actions $(\bmu_{0}^{*},\ldots,\bmu_{H-1}^{*})$ is also stochastic. 

\vspace{\baselineskip}

\begin{definition}
\thlabel{evtg_def}
    The \textit{expected-value-to-go (EVTG)} from time step $k+1$ onward after taking action $\bmu$ in state $\x_{k}$ is defined as
    \begin{align}
    &W_{k+1}^{*}(\x_{k},\bmu) \notag \\
    & = \max_{(\pi_{k+1},\ldots,\pi_{H-1})} \sum_{i=k+1}^{H-1}\mathbb{E}\left[ r_{i}\left( \x_{i},\pi_{i}(\x_{i}) \right) +  r_{H}(\x_{H}) | \x_{k},\bmu \right] \notag \\
    & \text{subject to} \quad  \x_{k+1} = \; h_{k}(\x_{k},\bmu,\w_{k}); \\
    & \hspace{1.8cm} \x_{i+1} =  \; h_{i}(\x_{i},\pi_{i}(\x_{i}),\w_{i}); \; \pi_{i}(\x_{i})  \in  \; \mathcal{U}_{i}(\x_{i}), \notag
    \end{align}
    
    % $$
    % W_{k+1}^{*}(\x_{k}^{*},\bmu) = \sum_{i = k+1}^{H-1} \mathbb{E}\left[ r_{i}(\x_{i}^{*},\bmu_{i}^{*}) | \x_{k}^{*},\bmu \right] + \mathbb{E}\left[ r_{H}(\x_{H}^{*}) \right],
    % $$
    for $k = 0,\ldots,H-2$, and 
    $$
    \begin{aligned}
    W_{H}^{*}(\x_{H-1},\bmu) & = \mathbb{E}\left[ r_{H}(\x_{H}) | \x_{H-1},\bmu \right] \\
    \text{subject to} \quad \x_{H} & = h_{H-1}(\x_{H-1},\bmu,\w_{H-1}). \\
    & 
    \end{aligned}
    $$
\end{definition}

\begin{remark}
    The EVTG quantifies the optimal expected cumulative reward from step $k+1$ to the end of the decision horizon, conditioned on being in state $\x_{k}$ and taking action $\bmu$ at step $k$. All subsequent policies and their associated actions are optimal with respect to the state at their respective time step. This is why the EVTG term is denoted with the superscript $^{*}$ in $W_{k+1}$.
    % All subsequence state and actions, denoted with superscript $^{*}$, are assumed to follow the optimal sequence of policies. 
    Although the exact evaluation of $W_{k+1}^{*}(\x_{k},\bmu)$ is generally intractable, it is useful as a theoretical construct.
\end{remark}
\vspace{\baselineskip}

% \begin{definition}
% \thlabel{appro_evtg_def}
% \textbf{Approximate expected-value-to-go:} Since $W_{k+1}^{*}(\x_{k},\bmu)$ is generally intractable to compute, we use $\hat{W}_{k+1}(\x_{k},\bmu)$ for  $k = 0,\ldots,H-1$ to denote the approximate expected-value-to-go after taking action $\bmu$ in state $\x_k$, which is computationally tractable. 
% % \end{definition}

% \begin{remark}
%     Unlike its exact counterpart $W_{k+1}^{*}(\x_{k},\bmu)$ in \thref{evtg_def}, the approximate EVTG $\hat{W}_{k+1}(\x_{k},\bmu)$ is not strictly defined by a sum of expected rewards. Instead, it serves as a computationally tractable approximation. Hence, no explicit form of analytical expression is assumed for $\hat{W}_{k+1}(\x_{k},\bmu)$. One exception is the last term $\hat{W}_{H}(\x_{H-1},\bmu)$. At the last step, it can be written as the same form as  $W_{H}^{*}(\x_{H-1},\bmu)$ in \thref{evtg_def}, since only a single transition remains and this transition is known. 
    
%     % Besides, $\hat{\x}_{k}$ is used in this definition since the optimal state path, denoted with superscript $^{*}$, is not expected to be visited when using the approximate EVTG. Throughout the paper, any quantity representing an approximation is denoted using a hat superscript $\hat{}$. 
% \end{remark}
% \vspace{\baselineskip}

\begin{definition}
\thlabel{exact_Q}
    The \textit{exact Q-value} from time step $k$ onward of taking action $\bmu$ in state $\x_{k}$ is defined as 
    \begin{equation}
    Q_{k}^{*}(\x_{k},\bmu) = r_{k}(\x_{k},\bmu) + W_{k+1}^{*}(\x_{k},\bmu),
    \end{equation}
    for $ k = 0,\ldots,H-1$.
\end{definition}

% \begin{definition}
% \thlabel{appro_Q}
%     The \textit{approximate Q-value} of taking action $\bmu$ in state $\x_{k}$ is defined as 
%     \begin{equation}
%     \hat{Q}_{k}(\x_{k},\bmu) = r_{k}(\x_{k},\bmu) + 
%     \hat{W}_{k+1}(\x_{k},\bmu), 
%     \end{equation}
%     for $ k = 0,\ldots,H-1$.
% \end{definition}

\begin{remark}
    Comparing to $W_{k+1}^{*}(\x_{k},\bmu)$ in \thref{evtg_def}, the Q-value includes an additional term: the immediate reward obtained at step $k$, which can be explicitly computed based on $\x_{k}$ and $\bmu$. To ensure a fair comparison in terms of horizon length, $Q_{k}^{*}$ should be compared with $W_{k}^{*}$, since both quantities represent the cumulative reward from step $k$ onward. This alignment is also reflected in the analysis presented in the following sections.  
\end{remark}

\vspace{\baselineskip}
\begin{definition}
\thlabel{value_fun_def}
    The value-to-go $V^{*}_{k}(\x_{k})$ in state $\x_{k}$ from time step $k$ onward is defined as: 
    \begin{align}
        & V^{*}_{k}(\x_{k})  = \max_{(\pi_{k},\ldots,\pi_{H-1})}  \sum_{i=k}^{H-1}\mathbb{E}\left[ r_{i}\left( \x_{i},\pi_{i}(\x_{i}) \right) \right]  + \mathbb{E}\left[ r_{H}(\x_{H}) \right] \notag\\
        & \text{subject to  }   \x_{i+1} = h_{i}(\x_{i},\pi_{i}(\x_{i}),\w_{i}); \; \pi_{i}(\x_{i}) \in \mathcal{U}_{i}(\x_{i}),
    \end{align} 
\end{definition}
~\\
for $k=0,\ldots,H-1$.

\begin{remark}
    The value-to-go $ V^{*}_{k}(\x_{k})$ yields the optimal objective value from time step $k$ onward, starting at state $\x_{k}$. Regardless of the value of state $\x_{k}$, all the subsequent policies $\pi_{i} \;(i = k,\ldots,H-1)$ are optimal, each selected based on the observed state at their respective time steps. 
\end{remark}

\vspace{\baselineskip}
Based on the above definitions, we can establish a connection between $W_{k}^{*}$ and $V_{k}^{*}$ as follows:
% \begin{align}
% \label{W_Q_link}
%     & W_{k}^{*}(\x_{k-1},\bmu_{k-1}) \notag \\
%     & = \max_{\bmu_{k} \in \mathcal{U}_{k}(\x_{k}) }\mathbb{E}
%     \left[  Q_{k}^{*}(\x_{k},\bmu_{k}) | \x_{k-1},\bmu_{k-1} \right] ,
% \end{align}
\begin{equation}
\label{W_V_link}
    W_{k}^{*}(\x_{k-1},\bmu_{k-1}) = \mathbb{E}
    \left[ V_{k}^{*}(\x_{k}) | \x_{k-1},\bmu_{k-1} \right],
\end{equation}
which is not hard to verify. Note that the expectation $\mathbb{E}$ operates with respect to $V_{k}^{*}$ in \eqref{W_V_link} since $\x_{k}$ is random. We will use \eqref{W_V_link} later in Section~\ref{sec:applications}.

%%%%%%%%%%%%%%%%%%%%%%%%%%%%%%%%%%%%%%%%%%%%%%%%%%%%%%%%%%%%%%
\subsection{Dynamic Programming with Bellman's Equation}
The optimal policy sequence $(\pi_{0}^{*},\ldots,\pi_{H-1}^{*})$ to problem \eqref{mdp_form} can be recursively obtained by \textit{Bellman's principle of optimality}, which states that: An optimal policy has the property that, whatever the initial state and initial decision are, the remaining decisions must constitute an optimal policy with regard to the state resulting from the first decision \cite{bellman1966dynamic}. 

More specifically, the Bellman's principle of optimality can be mathematically represented by the following dynamic programming equations: 
\begin{equation}
\label{bellman_optimal}
    \bmu_{k}^{*}  = \pi_{k}^{*}(\x_{k}^{*}) = \argmax_{\bmu \in \mathcal{U}_{k}(\x_{k}^{*})} Q_{k}^{*}(\x_{k}^{*},\bmu)
\end{equation}
for $k = 0,\ldots,H-1$, where $\x_{0}^{*} = \x_{0}$.

All the states and actions in equation \eqref{bellman_optimal} are denoted with superscript $^{*}$ because optimal actions are taken at all time steps and optimal state path $(\x_{0}^{*},\ldots,\x_{H-1}^{*},\x_{H}^{*})$ is visited following the optimal sequence of actions $(\bmu_{0}^{*},\ldots,\bmu_{H-1}^{*})$. 

In order to obtain the optimal sequence of actions (policies), we need to solve equation \eqref{bellman_optimal} and iterating backwards over indices $k = H-1, H-2, \ldots, 0$. However, the iteration procedure suffers from the curse of dimensionality and is computationally intractable for many problems \cite{bellman1966dynamic}. 

Additionally, based on the above dynamic programming equations and the definitions of $W_{k}^{*}$ and $Q_{k}^{*}$, we can have:
\begin{equation}
\label{link_QW_exact}
    \mathbb{E} \left[ Q_{k}^{*}(\x^{*}_{k},\bmu^{*}_{k})  | \x^{*}_{k-1}, \bmu^{*}_{k-1} \right] = W^{*}_{k}(\x^{*}_{k-1},\bmu^{*}_{k-1}),
\end{equation}
for $k = 0,1,\ldots,H-1$. In our discussion on ADP below, we will replace $W^*$ and $Q^*$ in equation \eqref{link_QW_exact} by their respective approximations $\hat{W}$ and $\hat{Q}$, in which case equation \eqref{link_QW_exact} only holds approximately. We will come back to this point later. 

\subsection{Approximate Dynamic Programming}
To obtain an approximate solution to equation \eqref{bellman_optimal}, we need to replace the EVTG term $W_{k+1}^{*}$ with an approximate quantity that is easy to compute, and then perform the ``$\argmax$" operator on the approximate quantity. This approach is called \textit{approximate dynamic programming (ADP)}. Next, we introduce the relevant concepts and notation for ADP. 

We use $\hat{W}_{k+1}(\x_{k},\bmu)$ for  $k = 0,\ldots,H-1$ to denote the \textit{approximate expected-value-to-go (EVTG)} after taking action $\bmu$ in state $\x_k$, which is computationally tractable. Unlike its exact counterpart $W_{k+1}^{*}(\x_{k},\bmu)$ in \thref{evtg_def}, the approximate EVTG $\hat{W}_{k+1}(\x_{k},\bmu)$ is not strictly defined by a sum of expected rewards. Instead, it serves as a computationally tractable approximation. Hence, no explicit form of analytical expression is assumed for $\hat{W}_{k+1}(\x_{k},\bmu)$. Its expression depends on how you design your approximation. One exception is the last term $\hat{W}_{H}(\x_{H-1},\bmu)$. At the last step, it can be written as the same form as  $W_{H}^{*}(\x_{H-1},\bmu)$ in \thref{evtg_def}, since only a single transition remains and this transition is known. Similarly, we use $\hat{V}_{k}(\x_{k})$ for $k = 0,\ldots,H-1$ to denote the \textit{approximate value-to-go} in state $\x_{k}$, which is computationally tractable. As with $\hat{W}_{k+1}(\x_{k},\bmu)$, no particular analytical form is assumed for $\hat{V}_{k}(\x_{k})$.

We use $\hat{Q}_{k}(\x_{k},\bmu)$ to denote the \textit{approximate Q-value} of taking action $\bmu$ in state $\x_{k}$, which is written as
\begin{equation}
    \hat{Q}_{k}(\x_{k},\bmu) = r_{k}(\x_{k},\bmu) + 
    \hat{W}_{k+1}(\x_{k},\bmu), 
\end{equation}
for $ k = 0,\ldots,H-1$. Conceptually, replacing $W_{k+1}^{*}(\x_{k},\bmu)$ with $\hat{W}_{k+1}(\x_{k},\bmu)$ is equivalent to replacing $Q_{k}^{*}(\x_{k},\bmu)$ with $\hat{Q}_{k}(\x_{k},\bmu)$, since the immediate reward term $r_{k}(\x_{k},\bmu)$ can be explicitly computed. Similar to $Q^{*}_{k}$ in \thref{exact_Q}, $\hat{Q}_{k}$ should be compared with $\hat{W}_{k}$ to ensure a fair comparison in terms of horizon length.

\vspace{\baselineskip}
\begin{definition}
\thlabel{adp_policy_def}
    Given the approximate Q-value function, the ADP policy $\hat{\pi}_{k}$ and the associated action $\hat{\bmu}_{k}$ in state $\hat{\x}_{k}$ are defined as:
    \begin{equation}
        \hat{\bmu}_{k} = \hat{\pi}_{k}(\hat{\x}_{k}) = \argmax_{\bmu \in \mathcal{U}_{k}(\hat{\x}_{k})} \hat{Q}_{k}(\hat{\x}_{k},\bmu) 
    \end{equation}
for $k = 0,\ldots,H-1$, where $\hat{\x}_{0} = \x_{0}$ and $\hat{\x}_{k+1} = h_{k}(\hat{\x}_{k},\hat{\bmu}_{k},\w_{k})$. 
\end{definition}

\begin{remark}
    All the states and actions in \thref{adp_policy_def} are denoted with a hat superscript $\; \hat{} \;$, which represents approximations. This is because the optimal state path, denoted with $^{*}$, is not expected to be visited when using approximate EVTG. Throughout the paper, any quantity representing an approximation is denoted using a hat superscript $\; \hat{} \;$. 
\end{remark}

\vspace{\baselineskip}
A given approximate Q-value function, together with its ADP policy and associated action, is referred to as an \textit{ADP scheme}. Common examples of ADP schemes include model predictive control (MPC) \cite{grune2008infinite}, reinforcement learning \cite{watkins1989learning,suttonRLbook,silver2017mastering}, policy rollout \cite{bertsekas1999rollout}, etc. In MPC, people consider the system behavior over a finite prediction horizon that is shorter than the full horizon $H$. As a result, the exact Q-value function $Q_k^{*}$, which accounts for all future rewards up to the end of the horizon, is approximated by $\hat{Q}_k$, which only incorporates rewards over a truncated horizon. The approximate Q-value function $\hat{Q}_{k}$ can also be represented using function approximators such as neural networks. Different reinforcement learning algorithms then correspond to different training procedures for these neural-network-based $\hat{Q}_{k}$.

Given an ADP scheme, we generally have $$\mathbb{E} \left[ \hat{Q}_{k}(\hat{\x}_{k},\hat{\bmu}_{k})  | \hat{\x}_{k-1}, \hat{\bmu}_{k-1} \right] \neq \hat{W}_{k}(\hat{\x}_{k-1},\hat{\bmu}_{k-1})$$ due to approximation errors. If the ADP scheme is accurate enough such that $W_{k}^{*}(\x_{k-1}^{*},\bmu^{*}_{k-1}) \approx \hat{W}_{k}(\hat{\x}_{k-1}, \hat{\bmu}_{k-1})$, then we expect 
\begin{equation}
\label{link_QW_approx}
    \mathbb{E} \left[ \hat{Q}_{k}(\hat{\x}_{k},\hat{\bmu}_{k})  | \hat{\x}_{k-1}, \hat{\bmu}_{k-1} \right] - \hat{W}_{k}(\hat{\x}_{k-1},\hat{\bmu}_{k-1}) \approx 0. 
\end{equation}
Later, we use the expected value of the left-hand side of \eqref{link_QW_approx} as a measure of the approximation error in the ADP scheme.

\vspace{\baselineskip}
\begin{definition}
    The optimal value function of problem \eqref{mdp_form} is the 
    objective value evaluated under the optimal sequence of policies $(\pi_{0}^{*},\ldots,\pi_{H-1}^{*})$, denoted by $V^{*}$. 
    The approximate value function of problem \eqref{mdp_form} under an ADP scheme is the objective value evaluated under the sequence of ADP policies $(\hat{\pi}_{0},\ldots,\hat{\pi}_{H-1})$, denoted by $\hat{V}$.
\end{definition}

\begin{remark}
\thlabel{V_star_with_initial}
    Since the initial state $\x_{0}$ is given, we have 
    $$V^{*} = Q_{0}^{*}(\x_0,\bmu_{0}^{*}) = V^{*}_{0}(\x_{0}),$$
    which is not hard to verify by the definition of $\bmu_{0}^{*}$ and $V^{*}_{0}(\x_{0})$. In addition, since the reward function is non-negative at each step, we always have $0 \leq \hat{V} \leq V^{*}$.
\end{remark}
\vspace{\baselineskip}

In the results in Section.~\ref{sec:main_results}, we restrict our attention to initial states $\x_{0}$ satisfying 
$$V^{*} = V^{*}_{0}(\x_{0}) > 0.$$ 
If we consider the initial states $\x_{0}$ where $V^{*} = V^{*}_{0}(\x_{0}) = 0$, then by \thref{V_star_with_initial}, it follows that $\hat{V}  =V^{*} = 0$. Thus, any ADP scheme produces the optimal solution, making the performance analysis trivial.

%%%%%%%%%%%%%%%%%%%%%%%%%%%%%%%%%%%%%%%%%%%%%%%%%%%%%%%%%%%%%%
\subsection{Previous Results}
To evaluate the quality of an ADP scheme, there has been interests in quantifying the distance between $\hat{V}$ and $V^{*}$. Common performance metrics include the difference metric $|\hat{V} - V^{*}|$ or the ratio metric $\hat{V} / V^{*}$. 

Bertsekas \cite{bertsekas2012dynamic,bertsekas2005dynamic} studies the bounds on the difference metric $|\hat{V} - V^{*}|$ for both finite horizon and infinite horizon problems. His results are derived for general ADP schemes without restricting the approximation approach. There have also been efforts to provide bounds for some specific classes of ADP schemes. For instance, Gr\"{u}ne and Rantzer \cite{grune2008infinite} develop difference bounds for MPC when approximating infinite horizon problems under different conditions. Hertneck \emph{et al.}\cite{hertneck2018learning}~study the difference bound between learning-based approximate MPC policies and robust MPC policies. 

Although difference bounds provide useful performance guarantees, they may not accurately capture the quality of approximation when the magnitudes of $\hat{V}$ and $V^{*}$ are either very large or very small. Consider the simple example where $\hat{V} = 0.9 \times 10^{6}$ and $V^{*} = 10^{6}$. The absolute error is $|\hat{V} -  V^{*}| = 10^{5}$, which appears large in magnitude. However, the corresponding performance ratio is $\hat{V} / V^{*} = 90\%$, indicating a tight approximation. Similarly, the pairs $(\hat{V}_{1} = 9 , V^{*}_{1} =10)$ and $(\hat{V}_{2} = 99, V^{*}_{2} =100)$ have the same absolute error, i.e., $| \hat{V}_{1} - V^{*}_{1}| = | \hat{V}_{2} - V^{*}_{2}| = 1$. But their performance ratios differ substantially, with $\hat{V}_{1} / V^{*}_{1} = 90\%$ and $\hat{V}_{2} / V^{*}_{2} = 99\% $, respectively. Therefore, ratio bounds can provide complementary information to difference bounds and may offer a more informative characterization of approximation quality in certain settings.

Inspired by various curvature-based ratio bounds in submodular optimization, Liu \emph{et al.}\cite{liuadplcss}~derive a ratio bound for $\hat{V} / V^{*}$ in terms of the \textit{total curvature} and \textit{forward curvature}. However, both curvature quantities are intractable to compute, which makes the corresponding ratio bound difficult to obtain in practice.

%%%%%%%%%%%%%%%%%%%%%%%%%%%%%%%%%%%%%%%%%%%%%%%%%%%%%%%%%%%%%%
%%%%%%%%%%%%%%%%%%%%%%%%%%%%%%%%%%%%%%%%%%%%%%%%%%%%%%%%%%%%%%
\section{Main Results}
\label{sec:main_results}

%%%%%%%%%%%%%%%%%%%%%%%%%%%%%%%%%%%%%%%%%%%%%%%%%%%%%%%%%%%%%%
\subsection{Performance Guarantee}
\label{subsec:main_results_performance}
We are interested in deriving a computable ratio bound of the form $$\frac{\hat{V}}{V^{*}} \geq \beta,$$ for general ADP schemes. 

\vspace{\baselineskip}
\begin{lemma} 
\thlabel{equivalency_lemma}
Deriving a lower bound of the performance ratio is equivalent to deriving an upper bound of the optimal value function $V^{*}$.  
\end{lemma}

\begin{proof}
    By simple algebra, $\hat{V} / V^{*} \geq \beta$ if and only if $V^{*} \leq \hat{V} / \beta$. Furthermore, if an upper bound on $V^{*}$ is available, i.e., $V^{*} \leq \overline{V}$, then $\hat{V} / V^{*} \geq \hat{V} / \overline{V}$, which provides a lower bound on the performance ratio. 
\end{proof}
    
\vspace{\baselineskip}
The approximation error arises from the fact that the approximate Q-value $\hat{Q}_{k}$, with its associated ADP policy $\hat{\pi}_{k}$ and action $\hat{\bmu}_{k}$, is used at each step. The overall error will be viewed as the accumulation of stepwise errors, which can be decomposed in the expression of $V^{*}$ as follows: 
\begin{align}
\label{V*_decompose}
& V^{*}  = \sum_{k=0}^{H-1}\mathbb{E}\left[ r_{k}\left( \x_{k}^{*},\bmu_{k}^{*} \right) \right]  + \mathbb{E}\left[ r_{H}(\x_{H}^{*}) \right]  \\
& \hspace{-0.4em} \stackrel{(a)}{=} \sum_{k=0}^{H-1}\mathbb{E} \left[  r_{k} \left(\x_{k}^{*},\bmu_{k}^{*} \right) +  \hat{W}_{k+1}(\x_{k}^{*},\bmu_{k}^{*}) - \hat{W}_{k}(\x_{k-1}^{*},\bmu_{k-1}^{*}) \right]   \notag \\
& \hspace{-0.4em} \stackrel{(b)}{=} \hat{Q}_{0}(\x_{0},\bmu_{0}^{*}) + \sum_{k=1}^{H-1} \mathbb{E}\left[ \hat{Q}_{k}(\x_{k}^{*},\bmu_{k}^{*}) - \hat{W}_{k}(\x_{k-1}^{*},\bmu_{k-1}^{*}) \right]. \notag
\end{align}

Here, the approximate EVTG term $\hat{W}_{k}$ is a deterministic function of its state-action input. Once $\x_{k}^{*}$ and $\bmu_{k}^{*}$ are fixed, $\hat{W}_{k+1}(\x_{k}^{*},\bmu_{k}^{*})$ is fixed and can be regarded as already averaged. Therefore, the expectation $\mathbb{E}$ is averaging with respect to $r_{k}(\x_{k}^{*},\bmu_{k}^{*})$ term in $\hat{Q}_{k}(\x_{k}^{*},\bmu_{k}^{*})$. Equality (a) holds because of the telescoping sum, in which $\hat{W}_{0}$ is undefined and can be interpreted as 0, and $\hat{W}_{H}(\x_{H-1}^{*},\bmu_{H-1}^{*}) = \mathbb{E}\left[ r_{H}(\x_{H}^{*})\right]$. For equality (b), $r_{0}(\x_{0},\bmu_{0}^{*})$ can be taken outside the expectation $\mathbb{E}$ since $\x_{0}$ is given and $r_{0}(\x_{0},\bmu_{0}^{*})$ is fixed. Grouping $r_{0}$ and $\hat{W}_{1}$ terms yields $\hat{Q}_{0}$. 

In practice, both $\hat{Q}_{0}$ and each term within the summation after equality (b) in equation \eqref{V*_decompose} are intractable to compute when state-action inputs are denoted with superscript $^{*}$. To derive a computable upper bound on $V^{*}$, each of those individual terms can be upper bounded by a computable quantity. The details of the bounding theorem are as follows. 

\vspace{\baselineskip}
\begin{theorem}
\thlabel{bound_theorem}
    Assuming 
    \begin{itemize}
        \item $\bmu_{k}^{*} \in \mathcal{U}_{k}(\hat{\x}_{k})$ for $k = 0,\ldots,H-1$; 
        \item There exists a quantity $\epsilon_{k}$ such that, 
        $$\mathbb{E} \left[ \hat{Q}_{k}(\x_{k}^{*},\bmu_{k}^{*}) - \hat{W}_{k}(\x_{k-1}^{*},\bmu_{k-1}^{*}) \right] \leq \epsilon_{k}$$
        for $k = 1,\ldots,H-1$. 
    \end{itemize}
    Then, we have
    $$
    V^{*} \leq \hat{Q}_{0}(\x_{0},\hat{\bmu}_{0}) + \sum_{k = 1}^{H-1} \epsilon_{k} = \overline{V}, 
    $$
    and the ratio bound $\beta = \hat{V} / \overline{V} \leq \hat{V} / V^{*}.$
\end{theorem}

\begin{proof}
    Since $(\hat{\x}_{0},\hat{\x}_{1},\ldots,\hat{\x}_{H})$ is a realization of the state path when executing an ADP scheme, the set $\mathcal{U}_{k}(\hat{\x}_{k})$ represents the subset of feasible actions available at time step $k$ for that ADP scheme. Therefore, the assumption $\bmu_{k}^{*} \in \mathcal{U}_{k}(\hat{\x}_{k})$ guarantees that the optimal action $\bmu_{k}^{*}$ remains feasible under that ADP scheme at every step, which ensures that the decomposition in equation \eqref{V*_decompose} is valid. 
    
    Then, by combining equation \eqref{V*_decompose}, the definition of $\hat{\bmu}_{0}$ in \thref{adp_policy_def}, and the above assumption regarding $\epsilon_{k}$, we can obtain the upper bound $\overline{V}$. 
    
    By using the equivalency between the lower bound of the performance ratio and the upper bound of $V^{*}$ in \thref{equivalency_lemma}, we can obtain the ratio bound $\beta = \hat{V} / \overline{V} \leq \hat{V} / V^{*}.$
\end{proof}

\vspace{\baselineskip}
\begin{remark}
    In practice, computing $\bmu_{k}^{*}$ is often intractable, making it difficult to verify whether the assumption $\bmu_{k}^{*} \in \mathcal{U}_{k}(\hat{\x}_{k})$ holds. However, this assumption is automatically satisfied in certain cases, for instance, when the feasible control action set is independent of both time step and state. 
\end{remark}

\vspace{\baselineskip}
When dealing with a specific problem, we can set the quantity $\epsilon_{k}$ in \thref{bound_theorem} according to the requirement of that problem. One possible way to construct a computable $\epsilon_{k}$ is described as follows. We first set 
\begin{align}
\label{delta}
    &\delta_{k}(\x_{k-1},\bmu_{k-1}) = \mathbb{E} \left[ \left.\max_{\bmu_{k}\in \mathcal{U}_{k}(\x_{k})}  \hat{Q}_{k}(\x_{k},\bmu_{k}) \right| \x_{k-1},\bmu_{k-1} \right] \notag  \\
    & \hspace{2.7cm} - \hat{W}_{k}(\x_{k-1},\bmu_{k-1}) \notag \\
    & \text{subject to} \quad  \x_{k} = h_{k-1}(\x_{k-1},\bmu_{k-1},\w_{k-1}) \\
    & \qquad \quad \quad \bmu_{k-1} \in \mathcal{U}_{k-1}(\x_{k-1}), \notag
\end{align}
and 
\begin{equation}
\label{epsilon}
\epsilon_{k} = \max_{\substack{\x_{k-1} \in \mathcal{X} \\ \bmu_{k-1} \in \mathcal{U}_{k-1}(\hat{\x}_{k-1})}} \delta_{k}(\x_{k-1},\bmu_{k-1}),
\end{equation}
for $k = 1,\ldots,H-1$. The setup of $\delta_{k}(\x_{k-1},\bmu_{k-1})$ captures the stepwise approximation error induced by using approximate EVTG (approximate Q-value). This term characterizes the discrepancy between $\hat{Q}_{k}$ and $\hat{W}_{k}$, as formalized in equation \eqref{link_QW_approx}. Better ADP schemes, in general, lead to smaller magnitude of $\delta_{k}(\x_{k-1},\bmu_{k-1})$. Moreover, by conditioning $\hat{Q}_{k}$ on the distribution of $(\x_{k-1},\bmu_{k-1})$, the stepwise error $\delta_{k}(\x_{k-1},\bmu_{k-1})$ becomes an explicit function of the state-action pair $(\x_{k-1},\bmu_{k-1})$. This formulation enables us to maximize $\delta_{k}(\x_{k-1},\bmu_{k-1})$ over $(\x_{k-1},\bmu_{k-1})$ in order to obtain the upper bound $\epsilon_{k}$, which is also computationally tractable.

Next, each term within the summation after equality (b) in equation \eqref{V*_decompose} can be rewritten as:
\begin{align}
\label{stepwise_error_bound}
& \; \mathbb{E}\left[ \hat{Q}_{k}(\x_{k}^{*},\bmu_{k}^{*}) - \hat{W}_{k}(\x_{k-1}^{*},\bmu_{k-1}^{*}) \right]  \\ 
& \stackrel{(c)}{=} \mathbb{E}\left[ \mathbb{E} \left(  \left. \hat{Q}_{k}(\x_{k}^{*},\bmu_{k}^{*}) \right| \x_{k-1}^{*},\bmu_{k-1}^{*} \right)  - \hat{W}_{k}(\x_{k-1}^{*},\bmu_{k-1}^{*}) \right] \notag \\
& \stackrel{(d)}{\leq} \mathbb{E}\left[ \delta_{k}(\x_{k-1}^{*},\bmu_{k-1}^{*}) \right] \stackrel{(e)}{\leq} \mathbb{E}\left[ \epsilon_{k} \right] = \epsilon_{k}. \notag
% & \leq \mathbb{E}\left[ \epsilon_{k}(\x_{k-1}^{*},\bmu_{k-1}^{*}) \right] \leq \max_{\x_{k-1},\bmu_{k-1}} \epsilon_{k}(\x_{k-1},\bmu_{k-1}), \notag
\end{align}
The equality (c) in equation \eqref{stepwise_error_bound} is due to the law of total expectation. The inner expectation $\mathbb{E}$ is averaged over $\hat{Q}_{k}$, while the outer expectation integrates over the distribution of the state-action input $(\x_{k-1}^{*}, \bmu_{k-1}^{*})$. The inequalities (d) and (e) stem from the above setup of $\delta_{k}(\x_{k-1},\bmu_{k-1})$ and $\epsilon_{k}$. Lastly, $\epsilon_{k}$ is a constant after implementing the ``max'' operator on $\delta_{k}(\x_{k-1},\bmu_{k-1})$, which implies that the expectation sign $\mathbb{E}$ can be dropped in the last step of equation \eqref{stepwise_error_bound}.

Plugging equations \eqref{epsilon} and \eqref{stepwise_error_bound} into the conclusion in \thref{bound_theorem}, we can obtain: 
\begin{align}
\label{V*_bound_epsilon}
    V^{*} & \leq \hat{Q}_{0}(\x_{0},\hat{\bmu}_{0}) + \sum_{k = 1}^{H-1} \epsilon_{k} \\
    & = \hat{Q}_{0}(\x_{0},\hat{\bmu}_{0}) + \sum_{k=1}^{H-1} \max_{\substack{\x_{k-1} \in \mathcal{X} \\ \bmu_{k-1} \in \mathcal{U}_{k-1}(\hat{\x}_{k-1})}} \delta_{k}(\x_{k-1},\bmu_{k-1}). \notag
\end{align}
Well-designed ADP schemes generally yield small magnitudes of $\delta_{k}(\x_{k-1},\bmu_{k-1})$ and $\epsilon_{k}$. Moreover, such schemes typically ensure that $\hat{Q}_{0}(\x_{0},\hat{\bmu}_{0}) \gg \epsilon_{k}$ in \eqref{V*_bound_epsilon}, thereby leading to a tight upper bound of $V^{*}$. This will be reflected later in Section \ref{subsec:applications:LQG}. A special case arises when $\hat{W}_{k}$ is replaced by $W^{*}_{k}$ in equations \eqref{V*_decompose} and \eqref{delta}. In this setting, $\delta_{k}(\x_{k-1},\bmu_{k-1}) = \epsilon_{k} = 0$, and the decomposition of $V^{*}$ simply reduces to $V^{*} = Q^{*}_{0}(\x_{0},\bmu_{0}^{*})$.

\vspace{\baselineskip}
\begin{remark}
    Our proposed bounding framework in \thref{bound_theorem} involves the stepwise upper bound $\epsilon_{k}$. Certain situations provide $\epsilon_k$ without digital computation (e.g., an analytically derived $\epsilon_k$). In cases where we wish to compute $\epsilon_k$ numerically, we do not claim that the computation is always trivial. Equations \eqref{delta} and \eqref{epsilon} provide a general method to compute $\epsilon_k$ based on solving an optimization problem over a single time step. But this is not the only possible method. Depending on the ADP scheme, bounding the approximation error $\mathbb{E} [ \hat{Q}_{k}(\x_{k}^{*},\bmu_{k}^{*}) - \hat{W}_{k}(\x_{k-1}^{*},\bmu_{k-1}^{*}) ]$ might involve nonconvex or high-dimensional optimization. However, this optimization is much simpler than solving the original dynamic programming problem. Specifically, it involves the approximation error in a single time step and does not require a Bellman recursion or computing the optimal value-to-go function over the entire decision horizon.

    Furthermore, our result in \thref{bound_theorem} only requires an upper bound on the stepwise approximation error, rather than its exact value. Therefore, any optimization, relaxation, or approximation technique capable of producing a valid upper bound can be used under the proposed framework.
\end{remark}

\vspace{\baselineskip}
All the problem formulation and theoretical results developed so far focus on obtaining an upper bound $\overline{V}$ for $V^{*}$. If the objective in problem \eqref{mdp_form} is instead to minimize the expected cumulative cost, the same theoretical framework applies, with the only modification being that all ``max'' operators are replaced by ``min'' operators. In this case, we seek an lower bound $\underline{V}$ for $V^{*}$. One of the applications in the subsequent Section~\ref{sec:applications} deals with cost minimization and finding the lower bound $\underline{V}$.

%%%%%%%%%%%%%%%%%%%%%%%%%%%%%%%%%%%%%%%%%%%%%%%%%%%%%%%%%%%%%%
\subsection{Technical Distinction from Difference Bound Analysis}

\subsubsection{Distinction in Assumptions and Proof Techniques}
Although both the difference bound of Bertsekas in \cite{bertsekas2012dynamic} and the ratio bound developed in this paper characterize the gap between $\hat{V}$ and $V^{*}$, the underlying assumptions and proof techniques are fundamentally different. We highlight these differences in this subsection using our notation in this paper. 

In \cite{bertsekas2012dynamic}, a key assumption underlying Bertsekas' difference bound analysis is as follows: For all $k \text{ and } \x_{k}$, $| \hat{V}_{k}(\x_{k})-V^{*}_{k}(\x_{k}) | \leq \epsilon $. Under this assumption, Bertsekas first establishes the bound of stepwise policy loss 
\begin{equation}
\label{Bertsekas_stepwise_error}
V_{k}^{*}(\x_{k}) - Q^{*}_{k}(\x_{k},\hat{\bmu}_{k}) \leq 2\epsilon,
\end{equation}
and then accumulates the above stepwise loss over the decision horizon to obtain
\begin{equation}
\label{Bertsekas_difference_bound}
V^{*} - \hat{V} \leq 2H\epsilon. 
\end{equation}

In contrast, our framework our framework does not involve a given bound on the difference between the approximate and optimal value-to-go functions. Moreover, our stepwise approximation error is based on the quantity
$$\mathbb{E} \left[ \hat{Q}_{k}(\x_{k}^{*},\bmu_{k}^{*}) - \hat{W}_{k}(\x_{k-1}^{*},\bmu_{k-1}^{*}) \right],$$
as shown in \thref{bound_theorem}. Therefore, both the assumptions and proof techniques in this paper are different from the analysis in Bertsekas' difference bound.

\vspace{\baselineskip}
\subsubsection{Relationship Between Ratio and Difference Bounds} 
A ratio bound straightforwardly implies a corresponding difference bound, as follows. Suppose that $\hat{V} / V^{*} \geq \beta$. Then $V^{*} \leq \hat{V} / \beta$, which immediately yields
\begin{equation}
\label{ratio_to_diff}
V^{*} - \hat{V} \leq \frac{\hat{V}}{\beta} - \hat{V}. 
\end{equation}

Conversely, consider Bertsekas' difference bound in \eqref{Bertsekas_difference_bound}. Dividing both sides by $V^{*}$ gives
\begin{align}
\label{diff_to_ratio}
\frac{V^{*} - \hat{V}}{V^{*}} & = 1-\frac{\hat{V}}{V^{*}} \leq \frac{2H\epsilon}{V^{*}} ; \notag \\
\frac{\hat{V}}{V^{*}} & \geq 1- \frac{2H\epsilon}{V^{*}} \geq 1 - \frac{2H\epsilon}{\hat{V}}.
\end{align}
Compared with \eqref{ratio_to_diff}, the derivation in \eqref{diff_to_ratio} not only requires the uniform error bound $\epsilon$ for the approximate value-to-go function but also an additional lower bound for $V^{*}$, which is just $\hat{V}$ in \eqref{diff_to_ratio}. Consequently, the resulting ratio bound depends on multiple layers of relaxation and might be loose. Therefore, based on the less restrictive assumptions and the contrast between ratio and difference bounds, our framework is applicable to a broader class of ADP schemes. 
% This suggests that a ratio bound, in general, cannot be obtained naturally from an existing difference bound. 

%%%%%%%%%%%%%%%%%%%%%%%%%%%%%%%%%%%%%%%%%%%%%%%%%%%%%%%%%%%%%%
\subsection{Connection to Submodular Optimization}
\label{subsec:main_results_submodular}
We develop a ratio bound in one of the previous subsections, which resonates the seminal ratio bound for greedy algorithms in submodular optimization. In this subsection, we show that our ADP bounding framework generalizes those results in submodular optimization. 

\subsubsection{Existing Results}
Let $f: 2^{\mathcal{U}} \xrightarrow{}  \mathbb{R}_{\geq 0}$ be a set function satisfying the following conditions: 
\begin{itemize}
    \item \textbf{Null on null:} $f(\varnothing) = 0$.
    \item \textbf{Monotonicity:} For every $X \subseteq Y \subseteq \mathcal{U}$, we have $f(X) \leq f(Y)$. 
    \item \textbf{Submodularity:} For every $X \subseteq Y \subseteq \mathcal{U}$ and every $x \in \mathcal{U} \setminus Y$, we have $f(X \cup \{ x\}) - f(X) \geq f(Y \cup \{ x\}) - f(Y)$.
\end{itemize}
Such a function $f$ is called a monotone submodular set function. 

A classical problem in submodular optimization is: 
\begin{equation}
\label{sub_opt}
\begin{aligned}
    & \text{maximize } f(S), \text{ subject to } |S| \leq H \text{ and } S \subseteq \mathcal{U}.
\end{aligned}
\end{equation} 
Since $S$ is a set, it can be written as $S = \{s_{0},s_{1},\ldots,s_{H-1} \}$, starting from horizon index 0. The seminal result from Nemhauser et al. \cite{nemhauser1978analysis} shows that the greedy algorithm, which selects the element maximizing the incremental reward at each step, produces a solution whose objective value is at least a $(1-e^{-1})$-fraction of the optimal value, namely, 
$$\frac{f(G_{H-1})}{f(O_{H-1})} \geq \beta = 1-e^{-1},$$ 
where $G_{H-1} = \{g_{0}, g_{1},\ldots,g_{H-1}\}$ and $O_{H-1} = \{o_{0},o_{1},\ldots,o_{H-1}\}$ denote the solutions obtained by the greedy algorithm and the optimal policy up to horizon index $H-1$, respectively. The $\beta = (1-e^{-1})$ bound was later proven to hold when maximizing a monotone submodular function under a matroid constraint \cite{calinescu2011maximizing}. 

To improve upon the $\beta = (1-e^{-1})$ bound in certain cases, various notions of curvatures have been proposed, along with their corresponding performance bounds. One seminal and computable notion is the \textit{greedy curvature} proposed in \cite{conforti1984submodular}, which is defined as: 
$$
\gamma_{G} = \max_{1\leq k \leq H-1}\max_{ \bigl\{ s | \Delta \left( G_{k-1} \cup \{s\} \right)  > 0 \bigl\} } \Bigl\{ \frac{f(s)}{\Delta(G_{k-1} \cup \{s\})} \Bigl\},
$$
where $\Delta(G_{k-1} \cup \{s\}) = f(G_{k-1} \cup \{s\}) - f(G_{k-1})$ denotes the marginal gain by adding element $\{s\}$ to the greedy solution of horizon $k-1$. The performance bound associated with the greedy curvature is given by:
$$
\frac{f(G_{H-1})}{f(O_{H-1})} \geq \beta = \frac{1}{H} + \frac{1}{\gamma_{G}}\frac{H-1}{H}. 
$$

More recently, there have been efforts trying to extend the results on submodular set functions to submodular string (or sequence) functions, in which the order of the elements matters, and repeating elements are allowed in a string. In this setting, the set $S$ in problem \eqref{sub_opt} is replaced by a string that is written as $S = s_{0}s_{1}\cdots s_{H-1}$. Accordingly, the subset relation $\subseteq$ in the set case is replaced by the prefix relation $\preccurlyeq$, and the union operator $\cup$ is replaced by the concatenation operator $\oplus$, which appends elements or strings after an existing string.

Streeter and Golovin extend the $\beta = (1-e^{-1})$ bound to string submodular functions \cite{streeter2008online}. Subsequently, we extend the greedy curvature bound to string settings \cite{van2023improved}, and later show that a simple and computable bound is provably better than the greedy curvature bound \cite{li2024bounds,van2025performance}.

\vspace{\baselineskip}
\subsubsection{Generalization} 
In the setting of stochastic optimal control (or equivalently, Markov decision process) that we primarily study in this paper, different orders of the same set of control actions usually lead to different objective values. Based on this observation, we show how our ADP bounding framework generalizes the results from string submodular optimization. 

\textbf{Formulation under the MDP Framework:} Let $S_{i} = s_{0}s_{1}\cdots s_{i}$ for $i = 0,\ldots,H-1$, where $S_{i}$ denotes the substring consisting of the elements up to horizon index $i$ of a string $S$. Using the technique of telescoping sum, the objective function $f(S)$ can be decomposed into the sum of $H$ stepwise reward: 
\begin{equation}
\label{sub_decompose}
f(S) = \Delta(S_{0}) + \Delta(S_{1}) + \cdots + \Delta(S_{H-1}),
\end{equation}
where $\Delta(S_{i}) = f(S_{i}) - f(S_{i-1})$ for $i = 0,\ldots,H-1$ with $S_{-1} = \varnothing$.

With the decomposition, the submodular optimization problem in \eqref{sub_opt} can be reformulated as a simplified version of the stochastic optimal control: 
\begin{align}
\label{sub_mdp}
\max_{ (s_{0},s_{1},\ldots,s_{H-1}) } & \sum_{k = 0}^{H-1} \Delta(S_{k}) \notag \\
\text{subject to} \quad & \x_{k+1} = h_{k}(\x_{k},s_{k}), \; s_{k} \in \mathcal{U}_{k}(\x_{k}),  \\
& k = 0,\ldots,H-1. \notag
\end{align}
Here, the state variable is defined as $\x_{k} = S_{k-1}$, representing the sequence of actions taken up to the previous time step $k-1$, and $\mathcal{U}_{k}(\x_{k})$ denotes the feasible set of actions at time step $k$, given the previously selected actions. Unlike the general MDP formulation in \eqref{mdp_form}, the expectation operator $\mathbb{E}$ and the random variables $\w_{k}$ are omitted, since both the objective function and the transition dynamics are deterministic in the context of submodular optimization. 

\textbf{Greedy Algorithm and Classical Bound:} Under the simplified setting in \eqref{sub_mdp}, the quantities defined in Section~\ref{sec:preliminaries} can be expressed using the notation of submodular optimization. In particular, the stepwise reward is given by 
$$
r_{k}(\x_{k},s_{k}) = \Delta(S_{k}) \text{ for } k = 0,\ldots,H-1.
$$ 
The expected-value-to-go (EVTG) term can be written as: 
$$
W^{*}_{k+1}(\x_{k},s_{k}) = \max_{ \{S : |S| = H-k \}} \bigl[ f( \x_{k} \oplus s_{k}  \oplus S ) - f( \x_{k} \oplus s_{k} ) \bigl]
$$
for $k = 0,\ldots,H-2$.  Accordingly, the exact $Q$-value and the dynamic programming equations under Bellman's principle of optimality can be formulated in terms of the above $r_{k}$ and $W_{k+1}^{*}$.

The widely used greedy algorithm can be interpreted as a special kind of ADP scheme, in which the EVTG term $W^{*}_{k+1}(\x_{k},s_{k})$ is neglected, i.e., $\hat{W}_{k+1}(\x_{k},s_{k}) = 0$. The greedy algorithm selects actions by maximizing only the stepwise reward $\Delta(S_{k})$ (equivalently, the $r_{k}(\x_{k},s_{k})$ term) at each time step. Following the notation of $G_{H-1}$ and $O_{H-1}$ introduced previously, let $O_{H-1} = o_{0}o_{1}\cdots o_{H-1}$ denote the optimal sequence of actions for problem \eqref{sub_mdp}, with associated optimal objective value $V^{*}$, and let $G_{H-1} = g_{0} g_{1} \cdots g_{H-1}$ denote the sequence of actions obtained by the greedy algorithm, with associated objective value $\hat{V}$. It then follows that
$$
\frac{\hat{V}}{V^{*}} = \frac{f(G_{H-1})}{f(O_{H-1})} \geq \beta = 1-e^{-1}.
$$

\textbf{Our Bounding Framework:} In addition to the classical $\beta = (1-e^{-1})$ bound for the greedy algorithm, we can apply our bounding framework developed in Section~\ref{subsec:main_results_performance} to problem \eqref{sub_mdp} to further quantify the performance of greedy algorithm. 

When using greedy algorithm as a special case of ADP scheme, the approximate EVTG term $\hat{W}_{k+1}(\x_{k},s_{k}) = 0$. As a result, equation \eqref{V*_decompose} reduces to: 
\begin{align}
\label{V*_decompose_sub}
V^{*} & = r_{0}(\varnothing,o_{0}) + \sum_{k=1}^{H-1} \mathbb{E} \bigl[ r_{k}(O_{k-1},o_{k}) \bigl] \notag \\
& = \sum_{k=0}^{H-1}r_{k}(O_{k-1},o_{k}) = \sum_{k=0}^{H-1} \Delta(O_{k}), 
\end{align}
where $O_{-1} = \varnothing$ and $O_{k} = o_{0}\cdots o_{k}$ for $k = 0,\ldots,H-1$. The expectation operators are omitted in \eqref{V*_decompose_sub} because both the stepwise reward functions and the transition dynamics are deterministic. By applying \thref{bound_theorem}, we can compute an upper bound for each $r_{k}(O_{k-1},o_{k})$ term, and thereby obtain an overall upper bound on $V^{*}$, denoted by $\overline{V}$. Specifically, 
\begin{align}
& r_{0}(\varnothing,o_{0}) = f(o_{0}) \leq \max_{s \in \mathcal{U}_{0}(\varnothing)} f(s) = f(g_{0}); \\
& r_{k}(O_{k-1},o_{k}) =  \Delta(O_{k}) = f(O_{k}) - f(O_{k-1}) \notag \\
& \leq \max_{s \in \mathcal{U}_{k}(O_{k-1})} \bigl[ f(O_{k-1} \oplus s) - f(O_{k-1}) \bigl] \stackrel{(a)}{\leq} f(s_{k}^{*})
\end{align}
for $k = 1,\ldots,H-1$, where 
$$
s_{k}^{*} = \argmax_{s \in \mathcal{U}_{k}(O_{k-1})} \bigl[ f(O_{k-1} \oplus s) - f(O_{k-1})\bigl].
$$
Inequality (a) above follows from the submodularity property of $f$, while the remaining inequalities above follow from the definitions of the ``max'' operator. Consequently, 
\begin{equation}
\label{V_bar_bound_sub}
V^{*} \leq f(g_{0}) + \sum_{k=1}^{H-1} f(s_{k}^{*}) \leq \sum_{k=0}^{H-1} \max_{s\in \mathcal{U}_{k}(S_{k-1})} f(s) = \overline{V},
\end{equation}
where $S_{k-1}$ in \eqref{V_bar_bound_sub} denotes the string formed up to horizon index $k-1$, with each action selected via its corresponding max operator. 

% \begin{remark}
%     To strictly follow the bound computation in \thref{bound_theorem}, 
% \end{remark}

Intuitively, we select $H$ actions that are feasible to be concatenated as a string and that individually yield the $H$ largest reward at the first time step. The corresponding performance bound $\beta = f(G_{H-1}) / \overline{V}$ can be referred to as the \textbf{top-$H$ bound}. This result coincides with our bounding result for string submodular optimization in \cite{van2025performance}, thereby implying that our bounding framework developed in Section~\ref{subsec:main_results_performance} generalizes the result in submodular optimization. Moreover, our prior work shows that the bound $\overline{V}$ is provably better than the seminal greedy curvature bound mentioned earlier \cite{van2025performance}.

%%%%%%%%%%%%%%%%%%%%%%%%%%%%%%%%%%%%%%%%%%%%%%%%%%%%%%%%%%%%%%
%%%%%%%%%%%%%%%%%%%%%%%%%%%%%%%%%%%%%%%%%%%%%%%%%%%%%%%%%%%%%%
\section{Applications}
\label{sec:applications}

%%%%%%%%%%%%%%%%%%%%%%%%%%%%%%%%%%%%%%%%%%%%%%%%%%%%%%%%%%%%%%
\subsection{Robot Path Planning via Linear Quadratic Gaussian (LQG)}
\label{subsec:applications:LQG}

\subsubsection{Problem Setup}
Consider a simple robot path planning problem shown in Fig.~\ref{robot_path_planning}, where the objective is to steer a robot from a given starting position (origin) to a target location (destination) on a smooth two-dimensional plane. The robot motion obeys classical Newtonian dynamics.

Specifically, the robot starts from an initial state $\x_{0} = [ x_{0},0,y_{0},0]^{T}$, which specifies its initial position ($[x_0,y_0]$) and velocity ($[0,0]$ in this case) in the 2D surface, and is driven toward a desired final state $\x_{f} = [ x_{f},0,y_{f},0]^{T}$ over a horizon of $H$ control actions. Its velocity is decomposed into two components along the $X$ and $Y$ axes, denoted by $\dot{x}$ and $\dot{y}$, respectively. The control inputs applied to the robot consist of forces in the $X$ and $Y$ directions, denoted by $F_{x}$ and $F_{y}$. These forces produce corresponding accelerations $\ddot{x}$ and $\ddot{y}$ in each direction. Additionally, $m$ represents the mass of the robot. 
%and $b$ denotes the coefficient of friction on the 2D surface.

\begin{figure}[hbt!]
    \centering
    \includegraphics[height = 3.5cm]{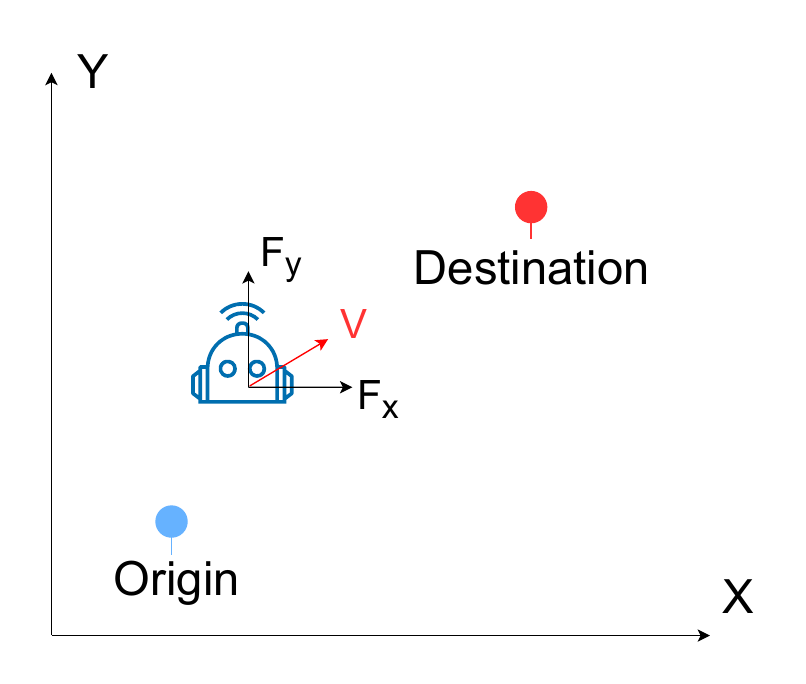}
    \caption{Robot path planning} 
    \label{robot_path_planning}
\end{figure}

Based on the above notation and Newton's second law of motion, the state-space representation of the continuous-time system dynamics is given by
% \begin{equation}
% \label{sys_dyn_cont}
%     F_{x} = m\ddot{x},\; F_{y} = m\ddot{y}.
% \end{equation}

\begin{equation}
\label{sys_dyn_cont}
{\setlength{\arraycolsep}{2pt}\renewcommand{\arraystretch}{0.9}
\begin{bmatrix}
\dot{x}\\
\ddot{x}\\
\dot{y}\\
\ddot{y}
\end{bmatrix} = 
\begin{bmatrix}
0 & 1 & 0 & 0 \\
0 & 0 & 0 & 0 \\
0 & 0 & 0 & 1 \\
0 & 0 & 0 & 0
\end{bmatrix}
\begin{bmatrix}
x\\
\dot{x}\\
y \\
\dot{y}
\end{bmatrix} + 
\begin{bmatrix}
0 & 0\\
1/m & 0\\
0 & 0 \\
0 & 1/m
\end{bmatrix}
\begin{bmatrix}
F_{x}\\
F_{y}
\end{bmatrix} + \w,
}
\end{equation}
where $\w$ is the white Gaussian process noise. 

Equation \eqref{sys_dyn_cont} can be discretized into $H$ time steps with a control action applied within each step and the duration of each time step being $T$. At each time step $k \; (k = 0,\ldots,H-1)$, we define the state vector as $\x_{k} = [x_{k},v_{x_{k}},y_{k},v_{y_{k}}]^{T}$ and the control action vector as $\bmu_{k} = [F_{x_{k}},F_{y_{k}}]^{T}$. The state vector captures the robot position $(x_{k},y_{k})$ and velocity $(v_{x_{k}},v_{y_{k}})$ in both $X$ and $Y$ directions, while the control action vector represents the forces applied in those respective directions. The discrete-time version of system dynamics is as follows: 
\begin{equation}
\label{state_space_dis}
\begin{aligned}
    \x_{k+1} & = A\x_{k} + B\bmu_{k} + \w_{k}, 
\end{aligned}
\end{equation}
%\vspace{\baselineskip}
\noindent where ${\setlength{\arraycolsep}{2pt}\renewcommand{\arraystretch}{0.9}
A=\begin{bmatrix}
1 & T & 0 & 0\\
0 & 1 & 0 & 0\\
0 & 0 & 1 & T\\
0 & 0 & 0 & 1
\end{bmatrix}, \;
B=\begin{bmatrix}
T^2/2m & 0\\ T/m & 0 \\ 0 & T^2/2m\\ 0 & T/m
\end{bmatrix},}$ and $\w_{k} \overset{i.i.d}{\sim} \mathcal{N}(\boldsymbol{0},\Sigma)$ is the process noise acting on the state evolution. 

\textbf{LQG Formulation:} Following the dynamics in equation \eqref{state_space_dis}, the robot path planning problem can be formulated under the linear quadratic Gaussian (LQG) control framework: 
\begin{align}
\label{LQG_ori}
    \argmin_{\pi_{0},\ldots,\pi_{H-1}} V &= \sum_{k = 0}^{H-1} \mathbb{E}  \left[ (\x_{k}-\x_{f})^{T}Q(\x_{k}-\x_{f}) + \bmu_{k}^{T}R\bmu_{k} \right] \notag \\
    & \quad \; + \mathbb{E} \left[ (\x_{H}-\x_{f})^{T}Q_{f}(\x_{H}-\x_{f}) \right] \notag \\
\text{s.t. } \x_{k+1} & = A\x_{k} + B\bmu_{k} + \w_{k},\;\; \bmu_{k} = \pi_{k}(\x_{k}) \\
\w_{k} &\overset{i.i.d}{\sim} \mathcal{N}(\boldsymbol{0},\Sigma). \notag
\end{align}
Since the goal is to drive the robot toward the target final state $\x_{f}$, the state cost at each step in \eqref{LQG_ori} can be characterized by the distance between the current state and the target final state. Specifically, the distances are modeled by the quadratic forms of $(\x_{k} - \x_{f})$ at each intermediate step and $(\x_{H} - \x_{f})$ at the terminal step, weighted by diagonal matrices $Q$ and $Q_f$, respectively. In addition, a control effort cost is included, defined as a quadratic form of control action $\bmu_{k}$, weighted by a diagonal matrix $R$. This cost structure penalizes both deviations from the target position and excessive control actions, which is very common in the LQG framework.

To solve problem \eqref{LQG_ori} using standard LQG approach, we introduce a variable substitution by setting $\z_{k} = \x_{k} - \x_{f}$ and $\z_{H} = \x_{H} - \x_{f}$. Then problem \eqref{LQG_ori} becomes: 
\begin{align}
\label{LQG_after}
    \argmin_{\pi_{0},\ldots,\pi_{H-1}} V &= \sum_{k = 0}^{H-1} \mathbb{E}  \left[ \z_{k}^{T}Q\z_{k} + \bmu_{k}^{T}R\bmu_{k} + \z_{H}^{T}Q_{f}\z_{H} \right] \notag \\
\text{s.t. } \z_{k+1} & = A\z_{k} + B\bmu_{k} + \dd + \w_{k},\;\; \bmu_{k} = \pi_{k}(\x_{k}) \\
\w_{k} &\overset{i.i.d}{\sim} \mathcal{N}(\boldsymbol{0},\Sigma), \notag
\end{align}
where $\dd = A\x_{f} - \x_{f} = \boldsymbol{0}$ in our setting. Therefore, the exact solution to problem \eqref{LQG_after} is given by the recursion characterized by the algebraic Riccati equation  as follows \cite{anderson2007optimal}:  
% $$
% \begin{aligned}
% P_{N} & = Q_f, \s_{N} = \boldsymbol{0}, c_{N} = 0 \\
% P_{k} & = Q + A^{T}P_{k+1}A - A^{T}P_{k+1}B(R+B^{T}P_{k+1}B)^{-1}B^{T}P_{k+1}A \\
% \s_{k} & = A^{T}(P_{k+1}\dd + \s_{k+1}) - A^{T}P_{k+1}B(R+B^{T}P_{k+1}B)^{-1}B^{T}(P_{k+1}\dd+\s_{k+1})\\ %+ Q\bs{x}_{f}
% c_{k} & =  c_{k+1} + \dd^{T}P_{k+1}\dd + 2\dd^{T}\s_{k+1} + \text{Tr}(\Sigma P_{k+1})\\ %+ \bs{x}_{f}^{T}Q\bs{x}_{f}
% & + (P_{k+1}\dd+\s_{k+1})^{T}B(R + B^{T}P_{k+1}B)^{-1}B^{T}(P_{k+1}\dd+\s_{k+1}) \\
% ~\\
% K_k & = (R + B^{T}P_{k+1}B)^{-1}B^{T}P_{K+1}A \\
% \bmu^{*}_{k} &= -K_{k}\z_{k} - (R + B^{T}P_{k+1}B)^{-1}B^{T}(P_{K+1}\dd+\s_{k+1})
% \end{aligned}
% $$
$$
\begin{aligned}
    & P_{H}  = Q_f, c_{H} = 0. \\
    & \text{For } k = H-1, H-2, \ldots, 0: \\
    & \quad S_{k+1} = R+B^{T}P_{k+1}B,\\
    & \quad P_{k}  = Q + A^{T}P_{k+1}A - A^{T}P_{k+1}BS_{k+1}^{-1}B^{T}P_{k+1}A,\\
    & \quad c_{k}  =  c_{k+1} + \text{Tr}(\Sigma P_{k+1}),\\
    & \quad K_k  = S_{k+1}^{-1}B^{T}P_{k+1}A,\\
    & \quad \bmu^{*}_{k} = -K_{k}\z_{k};\\
    & \text{End for.}
\end{aligned}
$$
\textbf{Value-to-go and EVTG Terms in LQG:} Based on the value-to-go in \thref{value_fun_def}, $V_{k}^{*}(\z_{k}) \; (k = 0,\ldots,H-1)$ in the LQG framework is defined as follows: 
\begin{align}
\label{V_value_fun_def}
& V_{k}^{*}(\z_{k}) = \min_{\pi_{k},\ldots,\pi_{H-1}} \sum_{i = k}^{H-1} \mathbb{E}  \left[ \z_{k}^{T}Q\z_{k} + \bmu_{k}^{T}R\bmu_{k} + \z_{H}^{T}Q_{f}\z_{H} \right] \notag \\
&\text{s.t. } \z_{i+1}  = A\z_{i} + B\bmu_{i} + \w_{i},\;\; \bmu_{i} = \pi_{i}(\x_{i}) \\ 
& \qquad \w_{i} \overset{i.i.d}{\sim} \mathcal{N}(\boldsymbol{0},\Sigma). \notag
\end{align}
The derivation of the above recursion also yields the expression for $V_{k}^{*}(\z_{k})$, given by
\begin{equation}
\label{V_value_fun}
    V_{k}^{*}(\z_{k}) = \z_{k}^{T}P_{k}\z_{k} + c_{k} = \z_{k}^{T}P_{k}\z_{k} + \sum_{i = k+1}^{H}\text{Tr}(\Sigma P_{i}).
\end{equation}
Moreover, we can obtain the analytical expression of the EVTG term $W_{k}^{*}(\z_{k-1}^{*},\bmu_{k-1}^{*})$ as described below. 

Given $\z_{k-1}^{*}$ and $\bmu_{k-1}^{*}$, the transition dynamics in problem \eqref{LQG_after} gives us $\z_{k}^{*} = A\z_{k-1}^{*} + B\bmu_{k-1}^{*} + \w_{k-1}$, which implies $\mathbb{E}(\z_{k}^{*}) = A\z_{k-1}^{*} + B\bmu_{k-1}^{*}$ and $\text{Cov}(\z_{k}^{*}) = \Sigma$. Thus, $W^{*}_{k}(\z^{*}_{k-1},\bmu^{*}_{k-1})$ can be computed analytically based on equation \eqref{W_V_link} from Section \ref{sec:preliminaries}, which links value-to-go and EVTG terms, and equation \eqref{V_value_fun}:
\begin{align}
\label{W_star_lqg}
    & \; W^{*}_{k}(\z^{*}_{k-1},\bmu^{*}_{k-1}) = \mathbb{E} 
    \left[ V_{k}^{*}(\z_{k}^{*}) | \z_{k-1}^{*},\bmu_{k-1}^{*} \right] \notag \\
    & = \mathbb{E} \left[ \left. \z_{k}^{*T}P_{k}\z_{k}^{*} + \sum_{i = k+1}^{H}\text{Tr}(\Sigma P_{i}) \right|  \z^{*}_{k-1},\bmu^{*}_{k-1} \right] \\
    % & = \mathbb{E}\left[\z_{k}^{*T}P_{k}\z_{k}^{*} \right] + \sum_{i = k+1}^{H}\text{Tr}(\Sigma P_{i}) \notag \\
    & \stackrel{(a)}{=} \text{Tr}(P_{k}\Sigma) + (A\z_{k-1}^{*} + B\bmu_{k-1}^{*})^{T}P_{k}(A\z_{k-1}^{*} + B\bmu_{k-1}^{*}) \notag \\
    & \quad \; + \sum_{i = k+1}^{H}\text{Tr}(\Sigma P_{i}), \notag
\end{align}
where the expression of $\mathbb{E}\left[\z_{k}^{*T}P_{k}\z_{k}^{*} \right]$ in equality (a) can be referenced in \cite{petersen2008matrix}.

In the following subsection regarding the experiments, we need to use the analytically computed EVTG quantities from equation \eqref{W_star_lqg} to serve as the labels of our imitation learning method, which will be explained in details shortly.

\vspace{\baselineskip}
\subsubsection{Experiments}
In our experiments, we set the robot dynamics parameters as: 
$$
\begin{aligned}
& \x_{0} = [0,0,0,0]^{T},\; \x_{f} = [100,0,100,0]^{T}, \\
& Q = \text{diag}(10,1,10,1), R = \text{diag}(0.5,0.5),  \\
& Q_{f} = \text{diag}(500,1000,500,1000), \Sigma = \text{diag}(5,2,5,2);
\end{aligned}
$$
and the control-related parameters as: 
$$
\begin{aligned}
    m = 1,\; T = 0.1, \text{ and } H=10.
\end{aligned}
$$
Essentially, the robot starts at rest from the position $[0,0]^{T}$ and is driven to park in the vicinity of target position $[100,100]^{T}$. To enforce the terminal accuracy, the magnitude of the terminal cost matrix $Q_{f}$ is set to be large, heavily penalizing the deviation from the desired final state $\x_{f}$. 

To obtain the ADP functions $\hat{Q}_{k}$ ($\hat{W}_{k}$) and the stepwise error $\epsilon_{k}$, we adopt an imitation learning approach, in which these functions are trained via supervised learning using expert demonstrations. In our setting, those expert demonstrations consists of optimal state-action pairs along with their associated EVTG labels. Imitation learning has been widely applied in autonomous driving \cite{pomerleau1988alvinn}, the training of AlphaGo \cite{silver2017mastering}, and many other sequential decision-making methods. 

In our experiments, we simulate $10^{6}$ optimal state paths generated by optimal sequence of actions obtained from the recursion characterized by the algebraic Riccati equation, in which the $10^{6}$ initial states are i.i.d.\ generated from $\mathcal{N}(\x_{0},\boldsymbol{\mathit{I}})$. These state paths, along with their corresponding action sequences, are partitioned into $H = 10$ clusters of input training datasets. 

% Note that the optimal solution to problem \eqref{LQG_after}, along with its associated quantities (the optimal sequence of states, the EVTG term $W_{k}^{*}$, and the optimal value function $V^{*}$), is computable and can be obtained analytically. Conducting experiments under the LQG framework allows us to directly compare the performance of an ADP scheme and its estimated ratio bound with the optimal value function and the true ratio bound, respectively. This provides a clear benchmark for evaluating the effectiveness of our bounding framework. Details regarding ADP function training and stepwise error function training are presented below.
\vspace{\baselineskip}
\begin{remark}
    Note that the optimal solution to problem \eqref{LQG_after}, along with its associated quantities (the optimal sequence of states, the EVTG term $W_{k}^{*}$, and the optimal value function $V^{*}$), can be obtained analytically. Therefore, from a practical standpoint, there is no need to employ an ADP scheme to solve an LQG-based problem. The purpose of our LQG example is as follows: Since the optimal solution is available, we can compare the estimated ratio bound produced by our framework with the true ratio bound, thereby demonstrating the effectiveness of the proposed framework. The details are shown below. 
\end{remark}

\vspace{\baselineskip}
\textbf{ADP Function Training:} To train the ADP functions, each cluster contains $10^{6}$ state-action pairs $(\z_{k}^{*},\bmu_{k}^{*})$ with the associated labels $W_{k+1}^{*}(\z_{k}^{*},\bmu_{k}^{*})$ for $k = 0,\ldots,H-1$. These state-action pairs with their labels serve as ``expert demonstrations'', since all the actions are generated by the optimal policy derived from the recursion. The quantity $W_{k+1}^{*}(\z_{k}^{*},\bmu_{k}^{*})$ can be computed exactly through equation \eqref{W_star_lqg}. We use a deep neural network (DNN) to perform supervised learning for each cluster, resulting in $H=10$ separately trained DNNs with each representing $\hat{W}_{k+1}(\z_{k},\bmu_{k})$. Given that 
$$
\hat{Q}_{k}(\z_{k},\bmu_{k}) = r(\z_{k},\bmu_{k}) + \hat{W}_{k+1}(\z_{k},\bmu_{k}),
$$ 
we can minimize $\hat{Q}_{k}(\z_{k},\bmu_{k})$ with respect to $\bmu_{k}$ for each $k = 0,\ldots,H-1$ to obtain the sequence of actions under our ADP scheme. 

%\vspace{\baselineskip}
\textbf{Stepwise Error Function Training:} To train the stepwise error functions, we also use DNNs to perform supervised learning to obtain $\delta_{k}(\z_{k-1},\bmu_{k-1})$ as shown in equation \eqref{delta}. In this case, the labels associated with each cluster are different from those in the training of the ADP functions. The first cluster contains $10^{6}$ state-action pairs $(\z_{0}^{*},\bmu_{0}^{*})$ with the associated labels $Q_{0}^{*}(\z_{0}^{*},\bmu_{0}^{*})$.
Each of the remaining clusters also contains $10^{6}$ state-action pairs $(\z_{k-1}^{*},\bmu_{k-1}^{*})$, but, with the associated labels 
$$
\mathbb{E} \left[  \left. Q_{k}^{*}(\z_{k}^{*},\bmu_{k}^{*}) \right| \z_{k-1}^{*},\bmu_{k-1}^{*} \right]  - W_{k}^{*}(\z_{k-1}^{*},\bmu_{k-1}^{*})
$$ 
for $k = 1,\ldots,H-1$. 

For the quantity $\mathbb{E} \left[  Q_{k}^{*}(\z_{k}^{*},\bmu_{k}^{*}) | \z_{k-1}^{*},\bmu_{k-1}^{*} \right]$, it is expressed as: 
$$
\begin{aligned}
& \mathbb{E} \left[  \left. Q_{k}^{*}(\z_{k}^{*},\bmu_{k}^{*}) \right| \z_{k-1}^{*},\bmu_{k-1}^{*} \right] \\
& = \mathbb{E} \left[ \left. r(\z_{k}^{*},\bmu_{k}^{*}) + W^{*}_{k+1}(\z_{k}^{*},\bmu_{k}^{*}) \right| \z_{k-1}^{*},\bmu_{k-1}^{*} \right]. 
\end{aligned}
$$
Given $(\z_{k-1}^{*},\bmu_{k-1}^{*})$, the next state $\z_{k}^{*}$ and its associated optimal action $\bmu_{k}^{*}$ are random variables with known means and variances. Thus, the quantity $\mathbb{E} \left[ r(\z_{k}^{*},\bmu_{k}^{*}) | \z_{k-1}^{*},\bmu_{k-1}^{*} \right]$ can be computed exactly. To estimate $\mathbb{E} \left[ W^{*}_{k+1}(\z_{k}^{*},\bmu_{k}^{*})  | \z_{k-1}^{*},\bmu_{k-1}^{*} \right]$, we generate $1000$ realizations of the state-action pairs $(\z_{k}^{*},\bmu_{k}^{*})$ conditioned on the fixed $(\z_{k-1}^{*},\bmu_{k-1}^{*})$, and average the 1000 values of $W^{*}_{k+1}(\z_{k}^{*},\bmu_{k}^{*})$. This empirical average serves as an approximation of $\mathbb{E} \left[ W^{*}_{k+1}(\z_{k}^{*},\bmu_{k}^{*}) | \z_{k-1}^{*},\bmu_{k-1}^{*} \right]$. Therefore, we have an overall approximation of the quantity $
\mathbb{E} \left[  Q_{k}^{*}(\z_{k}^{*},\bmu_{k}^{*}) | \z_{k-1}^{*},\bmu_{k-1}^{*} \right]  - W_{k}^{*}(\z_{k-1}^{*},\bmu_{k-1}^{*})
$. 

By using a DNN on each cluster to perform supervised learning, we obtain $\hat{Q}_{0}(\z_{0},\bmu_{0})$ and $\delta_{k}(\z_{k-1},\bmu_{k-1})$ for $k = 1,\ldots,H-1$.

% Minimizing $\hat{Q}_{0}(\z_{0},\bmu_{0})$ with respect to $\bmu_{0}$ together with minimizing $\epsilon_{k}(\z_{k-1},\bmu_{k-1})$ with respect to $\z_{k-1}$ and $\bmu_{k-1}$ will give us the lower bound of the optimal value function. 

\begin{remark}
    Ideally, we could train the $\mathbb{E} ( \hat{Q}_{k})$ and $\hat{W}_{k}$ terms in equation \eqref{delta} separately, and their difference would yield $\delta_{k}(\z_{k-1},\bmu_{k-1})$. However, in our implementation above, we directly approximate the entire difference term and use that approximation as the training label. Consequently, our trained DNNs can be interpreted as $\delta_{k}(\z_{k-1},\bmu_{k-1})$.
    %since $\epsilon_{k}(\z_{k-1},\bmu_{k-1})$ itself is an approximation quantity. 
    % $$
    % \max_{\bmu_{k}} \Bigl\{ \mathbb{E} \left(  \hat{Q}_{k}(\x_{k},\bmu_{k}) | \x_{k-1},\bmu_{k-1} \right)  - \hat{W}_{k}(\x_{k-1},\bmu_{k-1}) \Bigl\}
    % $$
\end{remark}

\vspace{\baselineskip}
For more details on our DNNs, the architecture of the DNN we use in both the training of ADP functions ($\hat{Q}_{k}$) and stepwise error ($\delta_{k}$) is presented in Table~\ref{DNN_arch}.

\begin{table}[h!]
    \centering
    \begin{tabular}{|c c |} 
     \hline
      & \thead{Architecture (Dimension and Operations)} \\ [0.5ex] 
     \hline\hline
     Input &  \thead{6}  \\ 
     Layer 1 & \thead{Linear: $6 \xrightarrow{} 128 + \text{BatchNorm} + \text{ReLU} + \text{Dropout}(0.1)$} \\ [1ex] 
     Layer 2 & \thead{Linear: $128 \xrightarrow{} 128 + \text{BatchNorm} + \text{ReLU} + \text{Dropout}(0.1)$} \\ [1ex] 
     Layer 3 & \thead{Linear: $128 \xrightarrow{} 64 + \text{BatchNorm} + \text{ReLU} + \text{Dropout}(0.1)$} \\ [1ex] 
     Output & \thead{Linear: $64 \xrightarrow{} 1$} \\ [1ex] 
     \hline
    \end{tabular}
    \caption{DNN Architecture}
    \label{DNN_arch}
\end{table}

\textbf{Comparison of $\hat{Q}_{0}$ and $\delta_{k}$:} We compare the above training labels for each stepwise error function, as shown in Fig.~\ref{Q_delta_compare}

\begin{figure}[hbt!]
    \centering
    \includegraphics[width = 0.8\columnwidth]{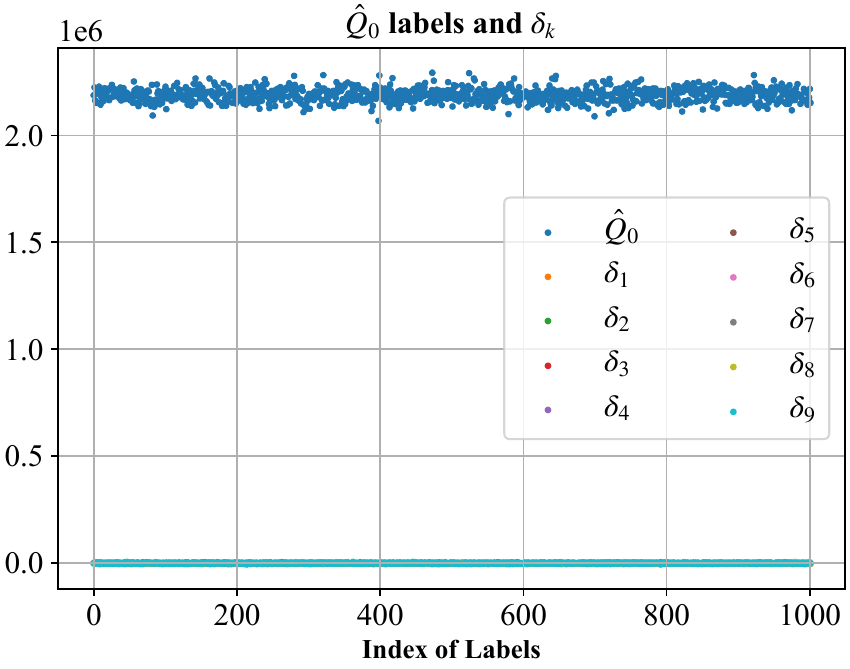}
    \caption{Comparison of training labels in stepwise error functions.} 
    \label{Q_delta_compare}
\end{figure}
We display the training labels of $\hat{Q}_{0}$ and $\delta_{k} \;(k = 1,\ldots,9)$, each with only 1000 random samples from a total of $10^{6}$. For each quantity, the remaining samples exhibit a similar numerical range as the 1000 displayed ones. The labels of $\hat{Q}_{0}$ are on the magnitude of $10^{6}$, whereas the $\delta_{k}$ terms are on the magnitude of $10^{3}$. Therefore, $\hat{Q}_{0} \gg \delta_{k}$, which aligns with the behavior expected from a well-designed ADP scheme, as discussed in Section \ref{subsec:main_results_performance}.

\textbf{Evaluation:} To evaluate our bounding framework, we generate another set of 100 i.i.d.\ test samples from $\mathcal{N}(\x_{0},\boldsymbol{\mathit{I}})$, i.e., 100 initial states $\x_{0}^{t}$ (or equivalently $\z_{0}^{t}$). We compare the optimal value function $V^{*}$, the approximate value function $\hat{V}$, and the lower bound $\underline{V}$ on these 100 test samples.  
\begin{itemize}
    \item \textbf{Optimal Cost $V^{*}$:} Using \thref{V_star_with_initial} and equation \eqref{V_value_fun}, we compute $V^{*} = V_{0}^{*}(\z_{0}^{t})$ for each of the 100 initial states. 

    \item \textbf{Approximate Cost $\hat{V}$:} To compute $\hat{V}$, we implement the ADP scheme by minimizing the trained $\hat{Q}_{k}(\z_{k},\bmu_{k})$ with respect to $\bmu_{k}$ for $ k = 0,\ldots,H-1$. Since $\hat{V}$ is an expected cost, we perform 500 independent rollout simulations for each initial state and take the empirical average as the estimate of $\hat{V}$. 

    \item \textbf{Lower Bound $\underline{V}$:} We minimize $\hat{Q}_{0}(\z_{0}^{t},\bmu_{0})$ and each $\delta_{k}(\z_{k-1},\bmu_{k-1})$ with respect to $\bmu_{0}$ and $(\z_{k-1},\bmu_{k-1})$, respectively. This yields $\hat{Q}_{0}(\z_{0}^{t},\hat{\bmu}_{0})$ and the corresponding quantities $\epsilon_{k}$ as shown in equation \eqref{epsilon} and \thref{bound_theorem}. Summing these terms produces the lower bound $\underline{V}$. 
\end{itemize}
 
\begin{figure}[hbt!]
    \centering
    \includegraphics[width = \columnwidth]{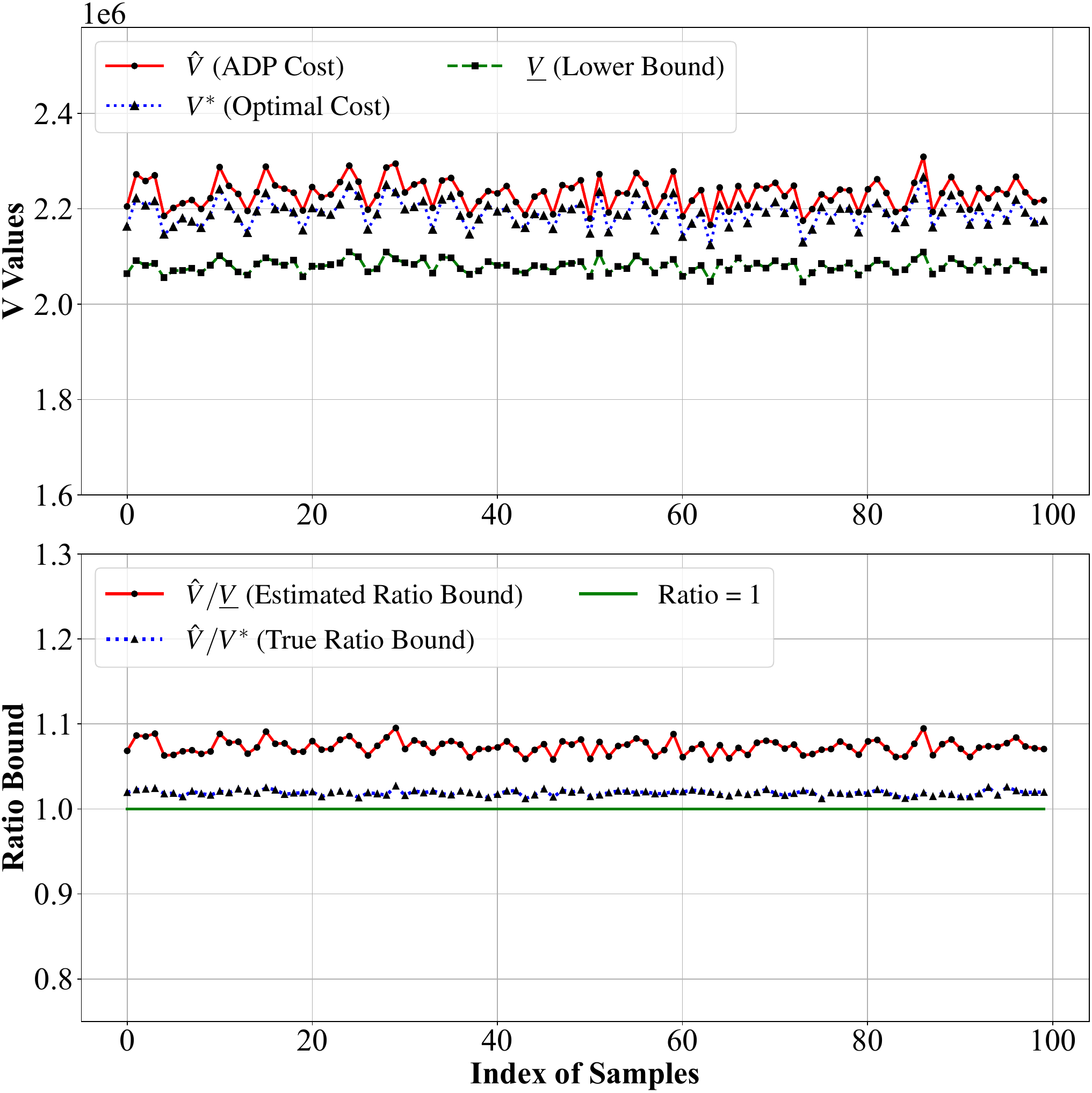}
    \caption{\textbf{Upper figure:} Comparison of $V^{*}$, $\hat{V}$, and $\underline{V}$; \; \textbf{Lower figure:} Comparison of the ratio bounds.} 
    \label{Ratio_Performance}
\end{figure}

At the end of Section~\ref{subsec:main_results_performance}, we state that our bounding framework applies to problems involving the minimization of expected cumulative cost, even though the main results in that section are derived under a maximization setting. Together with the fact that the LQG framework for our robot path planning formulates a cost minimization problem, this explains why minimization of the respective quantity is performed in the computation of the approximation cost ($\hat{V}$) and the lower bound ($\underline{V}$) above. Regarding the ratio bound, the estimated ratio $\hat{V}/\underline{V}$ serves as an upper bound for the true ratio $\hat{V}/V^{*}$, since $0 < \underline{V} \leq V^{*}$.

The results of our bounding framework are illustrated in Fig.~\ref{Ratio_Performance}. In the upper figure, as expected, the red line ($\hat{V}$) always lies above the blue line ($V^{*}$), while the green line ($\underline{V}$) remains the lowest. Moreover, the red and blue lines are very close to each other, indicating that the ADP scheme achieves near-optimal performance. 

In the lower figure, the red line (estimated ratio bound $\hat{V} / \underline{V}$) stays around 1.075, while the blue line (true ratio bound $\hat{V} / V^{*}$) stays around 1.025. This suggests that our framework provides a tight ratio bound that is close to the true performance.

%%%%%%%%%%%%%%%%%%%%%%%%%%%%%%%%%%%%%%%%%%%%%%%%%%%%%%%%%%%%%%
\subsection{Pendulum Stabilization}
In addition to the previous linear control example, we consider a nonlinear control problem based on the pendulum stabilization task. 

\subsubsection{Problem Setup} The objective is to stabilize a uniform-mass pendulum around the upright equilibrium position using bounded control inputs, as illustrated in Fig.~\ref{pendulum}.
\begin{figure}[hbt!]
    \centering
    \includegraphics[width=3.5cm]{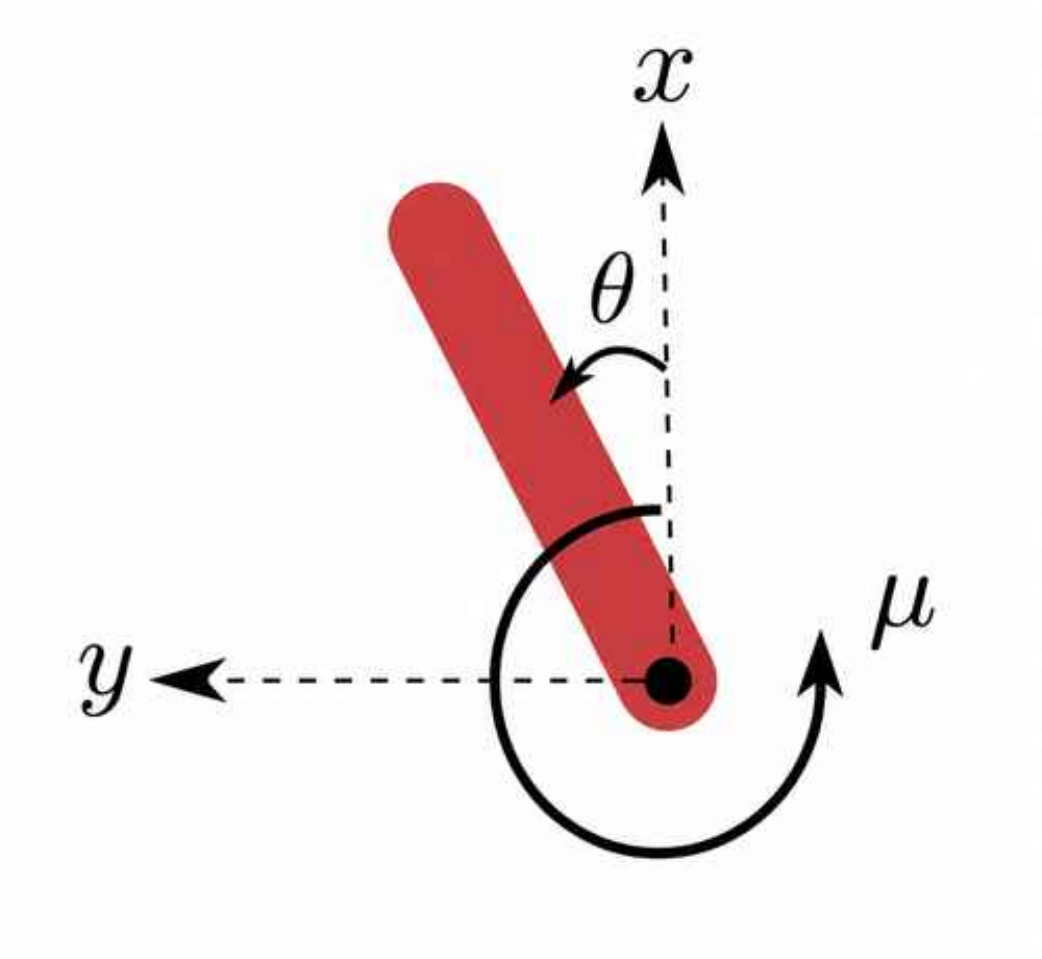}
    \caption{Pendulum stabilization task \cite{towers2024gymnasium}.} 
    \label{pendulum}
\end{figure}
The continuous-time frictionless pendulum dynamics are given by 
\begin{equation}
\label{pendulum_dyn_cont}
\ddot{\theta} = \frac{3g}{2l}\sin(\theta) + \frac{3}{ml^{2}}\mu,
\end{equation}
where $\theta$ denotes the pendulum angle, $\ddot{\theta}$ denotes the angular acceleration, $\mu$ is the applied torque, and $g$, $m$, and $l$ denote the gravitational acceleration, pendulum mass, and pendulum length, respectively \cite{towers2024gymnasium}.

Using a forward Euler discretization with sampling period $T$, the continuous-time dynamics can be discretized over a finite horizon of $H$ time steps as follows:
\begin{align}
\label{pendulum_dyn_dist}
    & \omega_{k+1} = \omega_{k} + T\left(\frac{3g}{2l}\sin(\theta_{k}) + \frac{3}{ml^{2}}\mu_{k}\right) \notag \\
    &  \theta_{k+1} =  \theta_{k} + T\omega_{k}
\end{align}
for each time step $k \; (k=0,\ldots,H-1)$, where $\theta_{k}$, $\omega_{k}$, and $\mu_{k}$ denote the pendulum angle, angular velocity, and applied torque at time step $k$, respectively. Defining the state vector as $\x_{k} = [\theta_{k},\omega_{k}]^{T}$, the discrete-time dynamics can be compactly expressed as $\x_{k+1} = h_{k}(\x_{k},\mu_{k})$, where $h_{k}$ denotes the state transition function defined by the discretized dynamics in equation \eqref{pendulum_dyn_dist}. Additionally, if the state deviates excessively from the upright equilibrium position, we consider the stabilization task to have failed. More precisely, if $|\theta_{k}| > \theta_{\text{max}}$ or $|\omega_{k}| > \omega_{\text{max}}$ for any time step $k$, then the task is terminated immediately, and only the rewards accumulated up to the termination time step are received. This is relevant to the training process using reinforcement learning (see below) because the received rewards play a role in tuning the neural-network weights.

Finally, the stabilization task is formulated as the following finite-horizon Markov decision process:
\begin{align}
\label{pendulum_mdp}
    \argmax_{(\pi_{0},\ldots,\pi_{H-1})} V &= \sum_{k=0}^{H-1} r_{k}(\x_{k},\mu_{k}) + r_{H}(\x_{H}) \notag \\ 
\text{s.t. } \x_{k+1} & = h_{k}(\x_{k},\mu_{k}),\;\; \mu_{k} = \pi_{k}(\x_{k}), \\
\mu_{k} & \in [-\mu_{\max},\mu_{\max}], \notag
\end{align}
where $r_{k}(\x_{k},\mu_{k}) = \exp(-3\theta_{k}^{2} - 3\omega_{k}^{2} - \mu_{k}^{2})$ and $r_{H}(\x_{H}) = \exp(-10\theta_{H}^{2} -\omega_{H}^{2})$ denote the intermediate stage reward and terminal reward functions, respectively. The terminal reward places large emphasis on maintaining the pendulum near the upright equilibrium position at the end of the horizon. 

\vspace{\baselineskip}
\subsubsection{Experiments} In our experiments, we set the pendulum system parameters as follows: 
$$
\begin{aligned}
    & T = 0.05 \; s, \;H = 20, \\
    & g = 10 \;m/s^{2}, \;m = 1kg, \; l = 1m, \\
    & \theta_{\max} = 0.35 \text{ rads}, \;\omega_{\max} = 1 \text{ rads/s}, \; \mu_{\max} = 1 \text{ Nm}.
\end{aligned} 
$$ 
The initial state $\x_{0}$ is sampled according to $\theta_{0} \sim \text{Unif}( -0.15, 0.15)$ and $\omega_{0} \sim \text{Unif} (-0.2,0.2)$. We train the controller by Gymnasium API \cite{towers2024gymnasium} using a reinforcement learning approach with proximal policy optimization (PPO) \cite{schulman2017proximal} to obtain the approximate policy $\hat{\pi}_{k}$ together with its corresponding approximate value-to-go function $\hat{V}_{k}$ for $k = 0,\ldots,H-1$. 

\textbf{Ratio Bound Computation:} Since only $\hat{V}_{k}$ is available after the training process, we set 
\begin{equation*}
    \begin{aligned}
        \hat{Q}_{0}(\x_{0},\hat{\mu}_{0}) & = \hat{V}_{0}(\x_{0}), \\
        \hat{W}_{k}(\x_{k-1},\mu_{k-1}) & = \hat{V}_{k}(h_{k-1}(\x_{k-1},\mu_{k-1}))
    \end{aligned}
\end{equation*} 
for $k = 1,\ldots,H-1$. We then compute the upper bound $\epsilon_{k}$ of the stepwise approximation error using equations \eqref{delta} and \eqref{epsilon} in Section~\ref{subsec:main_results_performance}. Note that the expectation operator $\mathbb{E}$ can be dropped when using equation \eqref{delta} due to the deterministic system dynamics in equation \eqref{pendulum_dyn_dist}. More specifically, we perform a grid search over the admissible state-action space $[-\theta_{\max},\theta_{\max}] \times [-\omega_{\max},\omega_{\max}] \times [-\mu_{\max},\mu_{\max}]$ to solve the maximization problems in equations \eqref{delta} and \eqref{epsilon}, thereby obtaining the upper bound $\epsilon_{k}$. Finally, applying \thref{bound_theorem} with the computed $\hat{V}_{0}(\x_{0})$ and $\epsilon_{k}$ yields the ratio bound $\beta$.

Three PPO policies are trained under the same computational budget, where each policy is trained under $3\times 10^{5}$ environment interactions using identical PPO hyperparameters except for the network size. For all three policies, both the policy network and value network consist of a single hidden layer with 16, 32, and 128 neurons, respectively. All three policies reach stable training performance before the estimated ratio bounds are computed. The corresponding training curves are shown in Fig.~\ref{pendulum_training}, where the rolling average curves clearly reach a plateau, indicating the convergence of the training processes. 

\begin{figure}[hbt!]
    \centering
    \includegraphics[width=\columnwidth]{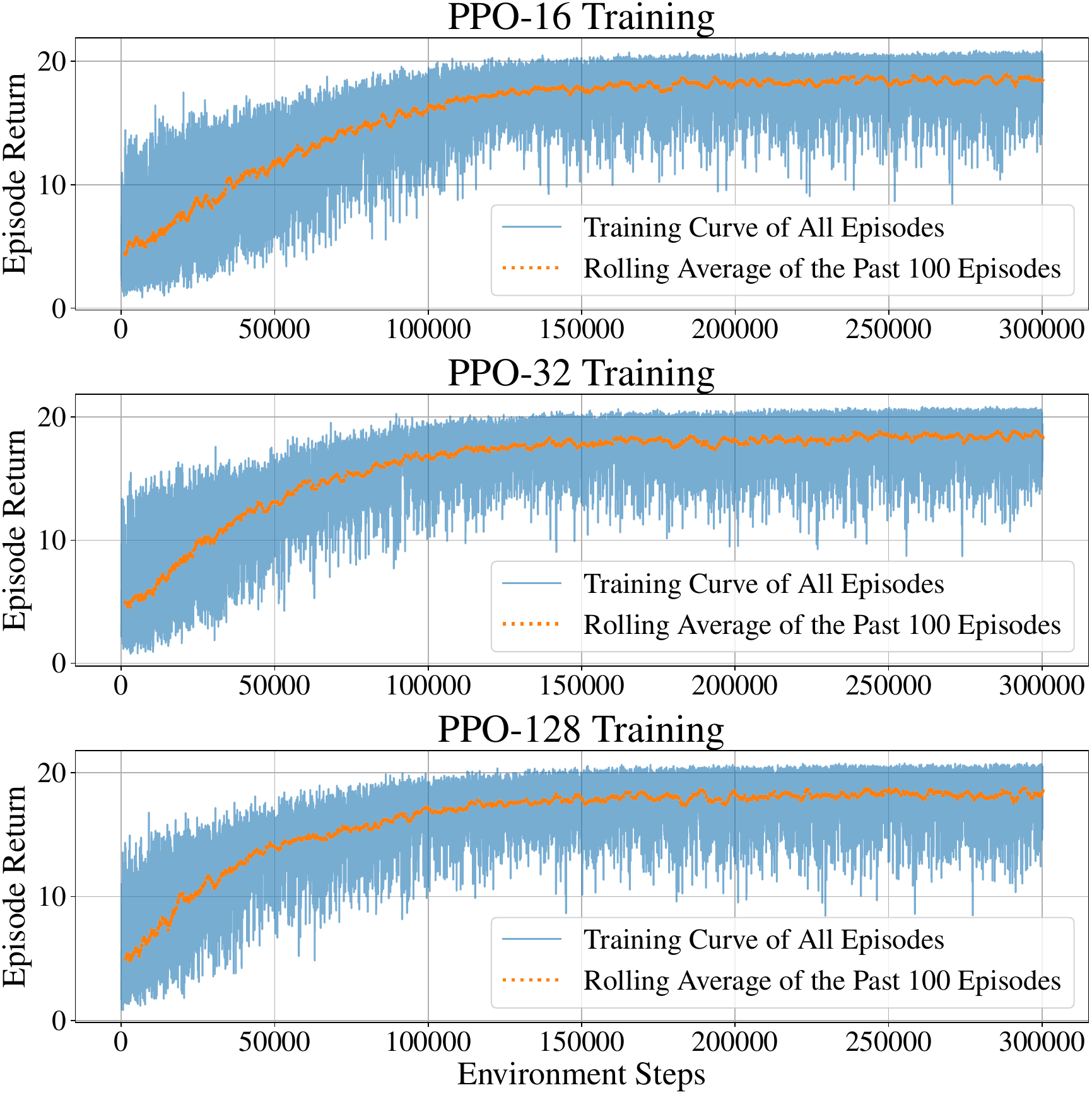}
    \caption{Training curves for three PPO policies.} 
    \label{pendulum_training}
\end{figure}

% To obtain the upper bound in \thref{bound_theorem}, we upper bound the first term by $\hat{Q}_{0}(\x_{0},\mu_{0}^{*}) \leq \hat{Q}_{0}(\x_{0},\hat{\mu}_{0})$ and each remaining term within the summation by
% \begin{equation}
% \label{stepwise_epsilon}
% r_{k}(\x^{*}_{k},\mu^{*}_{k}) + \hat{V}_{k+1}\left( h_{k}(\x^{*}_{k},\mu^{*}_{k}) \right) - \hat{V}_{k}(\x^{*}_{k}) \leq \epsilon_{k}
% \end{equation}
% for $k = 1,\ldots,H-1$. The quantity $\hat{Q}_{0}(\x_{0},\hat{\mu}_{0})$ can be obtained directly from the trained PPO policy. To compute the upper bound $\epsilon_{k}$, we perform a grid search over the admissible state-action space $[-\theta_{\max},\theta_{\max}] \times [-\omega_{\max},\omega_{\max}] \times [-\mu_{\max},\mu_{\max}]$ and select the state-action pair that maximizes the left-hand side of inequality \eqref{stepwise_epsilon}. 

For each trained PPO policy, we generate 50 sample trajectories that successfully completed the stabilization task and satisfies $|\theta_{H}| \leq 0.05$ and $|\omega_{H}| \leq 0.05$ at the terminal step (representing a terminal state relatively close to the upright equilibrium position).
We then compute the corresponding estimated ratio bounds through 
$$\frac{\hat{V}}{V^{*}} \geq \beta = \frac{\hat{V}} {\hat{V}_{0}(\x_{0}) + \sum_{k=1}^{H-1}\epsilon_{k}}.$$ 
The resulting estimated ratio bounds are presented in Fig.~\ref{pendulum_ratio_bound}. As the network size increases, the estimated ratio bound improves from approximately 75\% to 90\% and 95\%, indicating that more expressive functional approximations lead to tighter performance guarantees. Unlike the LQG-based robot path planning example in Section.~\ref{subsec:applications:LQG}, the exact value-to-go function and the optimal policy are unavailable for this nonlinear control problem. Therefore, only the estimated ratio bounds are reported. 

\begin{figure}[hbt!]
    \centering
    \includegraphics[width=\columnwidth]{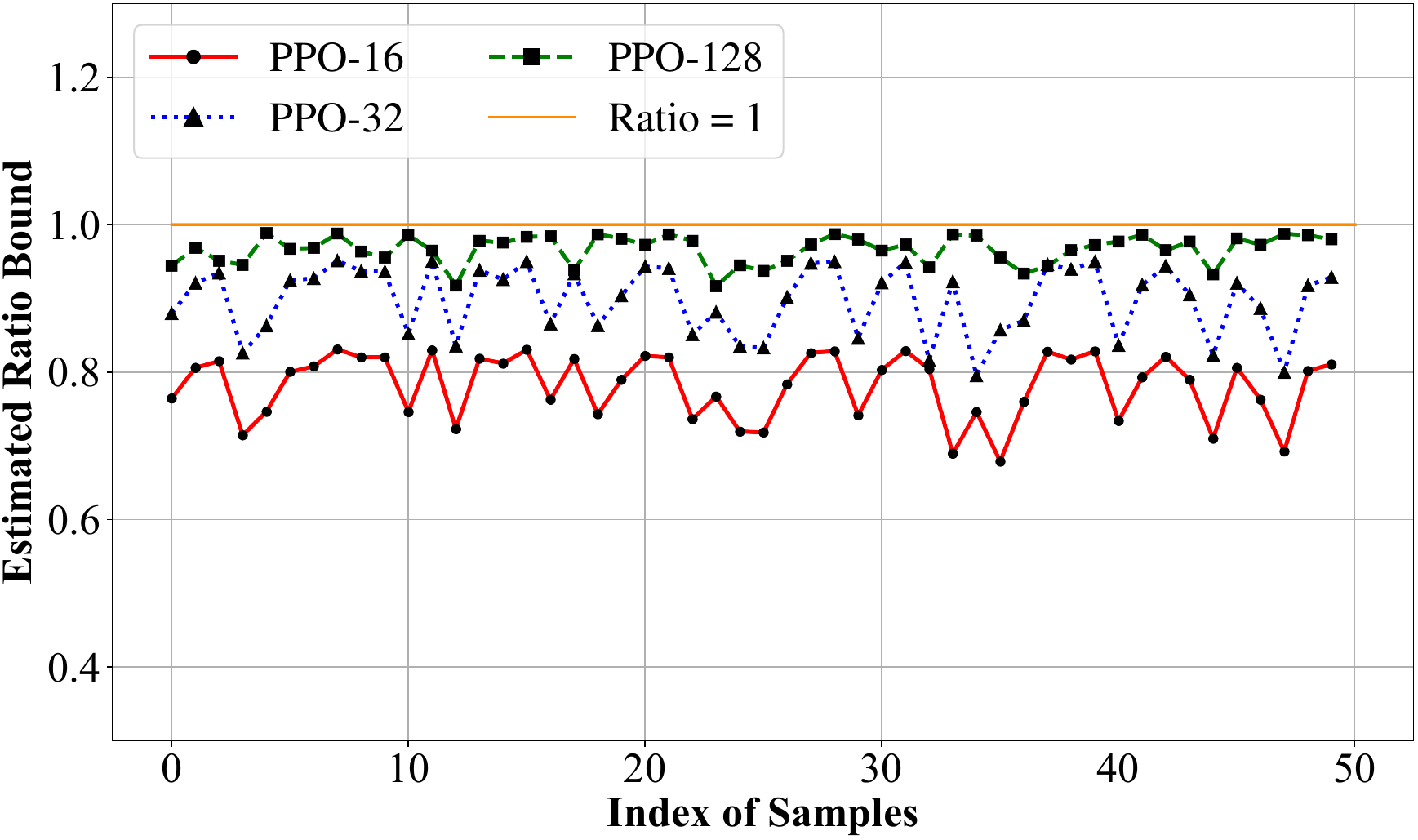}
    \caption{Estimated ratio bounds for pendulum stabilization task.} 
    \label{pendulum_ratio_bound}
\end{figure}

%%%%%%%%%%%%%%%%%%%%%%%%%%%%%%%%%%%%%%%%%%%%%%%%%%%%%%%%%%%%%%
\subsection{Multi-agent Sensor Coverage}
The multi-agent sensor coverage problem is originally studied in \cite{zhong2011distributed} and later analyzed in \cite{sun2019exploiting}. We also use this example to demonstrate our theoretical results regarding string submodular optimization in \cite{van2023improved},\cite{li2024bounds},\cite{van2025performance}. In this subsection, we apply our bounding framework to this example and show that the bound produced by our bounding framework coincides with the bound we developed in \cite{van2025performance}.

\subsubsection{Problem Setup}
We use the same way as in \cite{van2025performance} to describe the problem setup except for some slight notational changes. The mission space $\Omega \subset \mathbb{R}^{2}$ is modeled as a non-self-intersecting polygon where a total of $H$ sensors will be placed to detect a randomly occurring event. Those lattice points feasible for sensor placement within the mission space are denoted by $\Omega^{F} \subset \mathbb{R}^{2}$. The placement of these sensors occurs over $H$ steps, wherein one sensor is placed at each step. Our goal is to maximize the overall detection performance in the mission space, as illustrated in Fig.~\ref{sensors1}. 

\begin{figure}[hbt!]
    \centering
    \includegraphics[width=0.98\columnwidth]{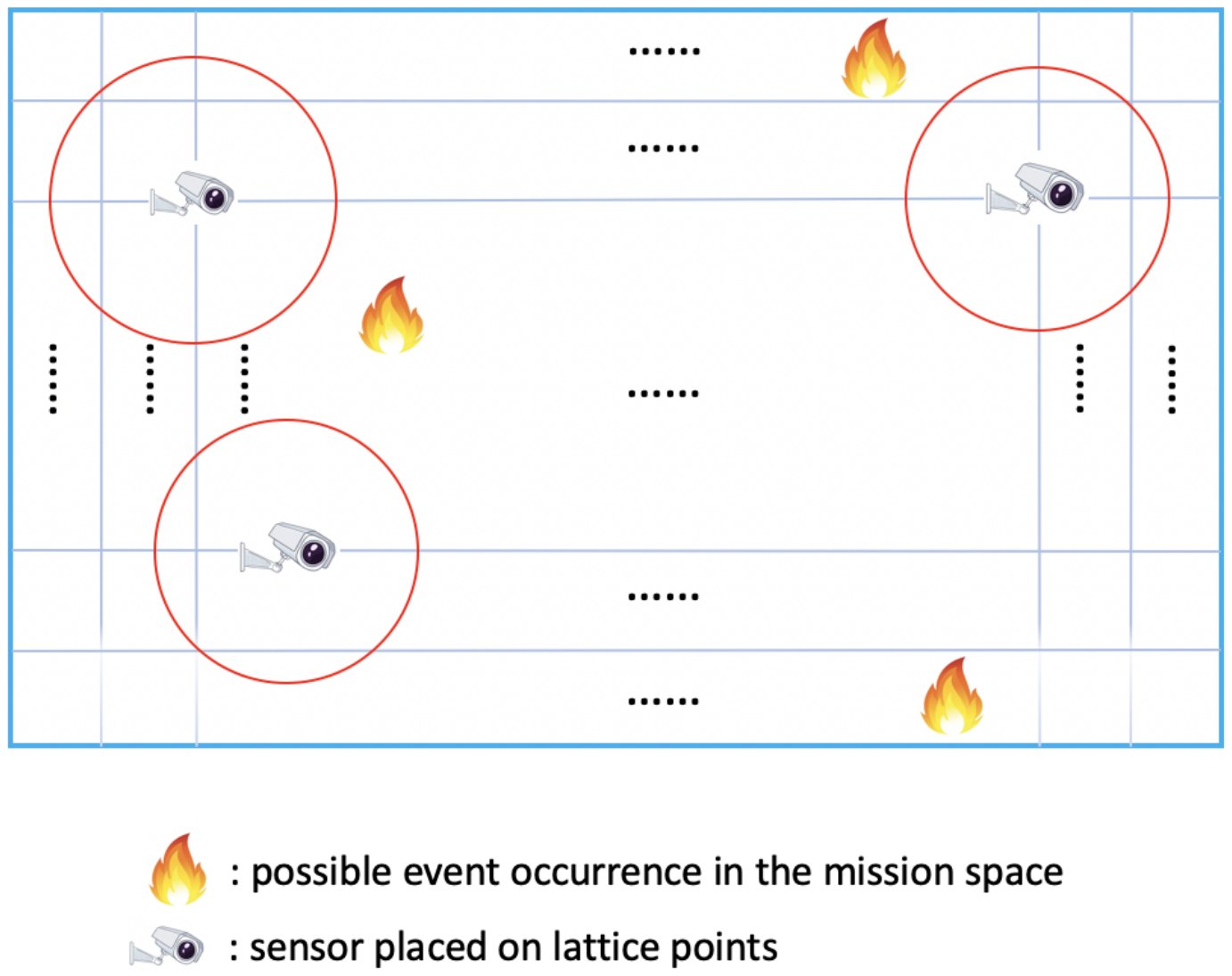}
    \caption{Sensor coverage for event detection in a mission space.} 
    \label{sensors1}
\end{figure}

The location of a sensor placed at step $k$ is denoted by $\mathbf{s}_{k}$, and it can detect any occurring event at location $\mathbf{x} \in \Omega$ with probability $p_{k}(\mathbf{x},\mathbf{s}_{k}) = \exp{(-\lambda_{k} \|\mathbf{x} - \mathbf{s}_{k}\|)}$. We assume that $\lambda_{i} > 0$, where $\lambda_i$ is the coverage decay rate of that sensor at step $k$ for $k = 0, \ldots, H-1$. Once all $H$ sensors have been placed, we represent the choice of these sensors with the string $S = (\mathbf{s}_0,\mathbf{s}_1,\ldots, \mathbf{s}_{H-1}) \in \left(\Omega^{F}\right)^H$. 

Moreover, we assume that all the sensors are working independently, and the decay rate can change depending on the step at which the sensor is placed. The sensing capability of all unused sensors may become worse with each passing step, and once the sensors are placed, their decay rates will remain constant from their placement step onward. Specifically, the decay rate variation is modeled as $\lambda_{k} = \lambda_{0} + \zeta t_{k}$ for $k = 0, \ldots, H-1$, where $\{ t_{k} \}$ is a monotonically increasing time index sequence beginning at $0$, and $\zeta$ is the parameter controlling the increase in the decay rate. Larger decay rate represents worse sensing capability. 

We can thus calculate the probability of detecting an occurring event at location $\mathbf{x} \in \Omega$ after placing $H$ sensors at locations $S$ using the following formula 

\begin{equation}
    \begin{aligned}
    P(\mathbf{x},S) & = 1-\prod_{k=0}^{H-1} \left( 1-p_k(\mathbf{x},\mathbf{s}_k) \right) \\
    & = 1-\prod_{k=0}^{H-1} \left( 1- e^{-\lambda_{k} \|\mathbf{x}-\mathbf{s}_k\|} \right).
    \end{aligned}
\end{equation}

To calculate the detection performance over the entire mission space $\Omega$, we incorporate the event density function $R$. The event occurrence over $\Omega$ is characterized by an event density function $R: \Omega \xrightarrow{} \mathbb{R}_{\geq 0}$, and we assume that $\int_{\mathbf{x} \in \Omega} R(\mathbf{x}) d\mathbf{x} < \infty$. Our objective function then becomes $H(S) = \int_{\mathbf{x} \in \Omega} R(\mathbf{x})P(\mathbf{x},S) d\mathbf{x}$, and we have the following string optimization problem:
\begin{equation}
\begin{aligned}
\label{obj_fun_sensor}
    \text{maximize } H(S), \; \text{ subject to } S \in \left(\Omega^{F}\right)^H.
\end{aligned}
\end{equation}

When the decay rate control parameter $\zeta > 0$ and the decay rates $\lambda_{k}$ vary across steps, the value of $H(S)$ depends on the order of the selected locations $\s_{k}$ in the string $S$. In this case, the sensors are \textbf{nonhomogeneous}, and each lattice point can be used more than once for sensor placement. 

A special case arises when $\zeta = 0$ and all decay rates are identical, i.e, $\lambda_{k} = \lambda_{0}$ for $k = 0,\ldots,H-1$. In this case, the value of $H(S)$ becomes independent of the order of $\s_{k}$, and the problem in \eqref{obj_fun_sensor} reduces to a set optimization problem. We refer to the sensors in this case as \textbf{homogeneous}. Note that each lattice point can only be used once in this case because no repeated elements is allowed in a set. 

It is shown in \cite{van2025performance} that $H(S)$ is a monotone string submodular function. Therefore, all theoretical results regarding set and string submodular functions can be applied to our multi-agent sensor coverage problem. 

\vspace{\baselineskip}
\subsubsection{Experiments} 
In our experiments, we consider a rectangular mission space $\Omega \subset \mathbb{R}^2$ of size $50 \times 40$, in which $H=5$ sensors will be placed. The whole mission space is partitioned into four identical regions as shown in Fig.~\ref{sensors2}. The top right and bottom left regions have high event densities whose $R(\mathbf{x})$ for each $\mathbf{x}$ is randomly generated from $\text{Unif}(0.5,0.8)$. The other two regions have low event densities whose $R(\mathbf{x})$ for each $\mathbf{x}$ is randomly generated from $\text{Unif}(0.1,0.3)$. The above setup reflects a scenario that random events are more likely to occur in the top right and bottom left regions. 

For the string case with nonhomogeneous sensors, we set the decay rate control parameter $\zeta = 0.1$, define the time index sequence as $t_{k} = 0.1k$ for $k = 0, \ldots, H-1$, and vary initial decay rate $\lambda_{0}$ over the range $\left[ 0.1,1.5 \right]$. For the set case with homogeneous sensors, we set $\zeta = 0$, and similarly vary the idenitical decay rate $\lambda$ over the range $\left[ 0.1,1.5 \right]$. 

\begin{figure}[hbt!]
    \centering
    \includegraphics[width=0.9\columnwidth]{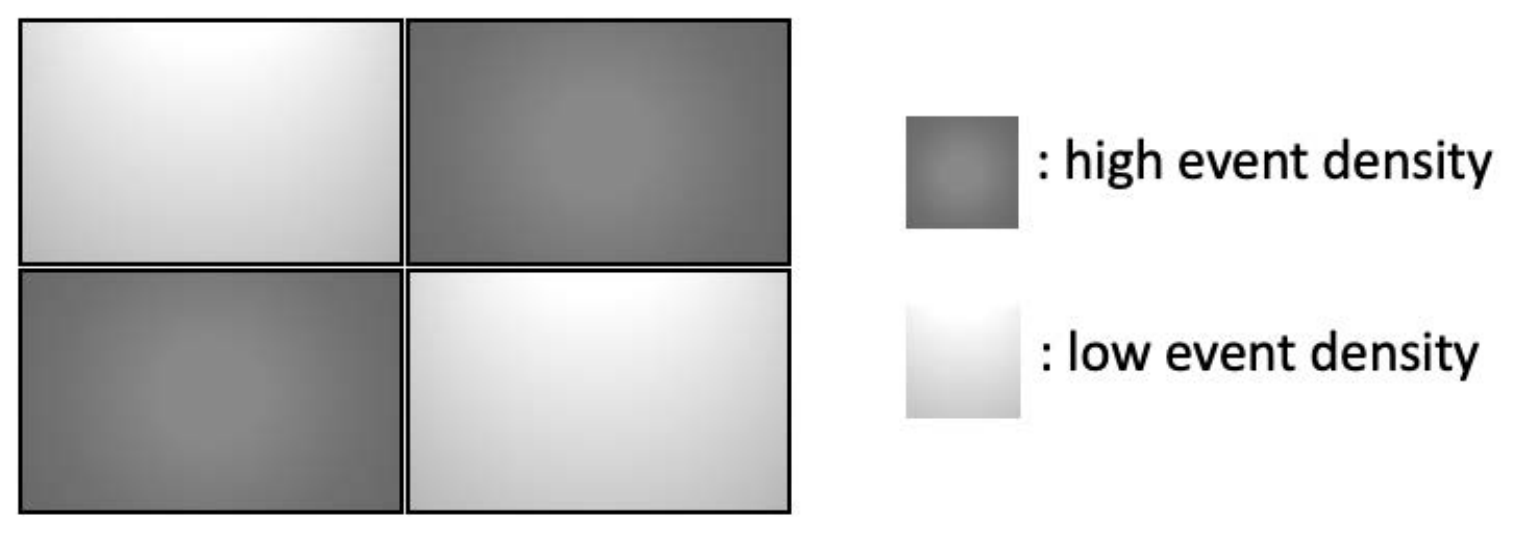}
    \caption{Mission space partition.} 
    \label{sensors2}
\end{figure}

We aim to compute and compare different performance bounds for the solution $G_{H-1}$ produced by the greedy algorithm. Let $\beta_{0}$ denote the classical $(1 - e^{-1})$ performance bound for set and string submodular functions, and let $\beta_{1}$ denote the computable greedy curvature bound in \cite{conforti1984submodular} and its generalized version in \cite{van2025performance}. 

Based on the theoretical results presented in Section~\ref{subsec:main_results_submodular}, we can compute the upper bound of the optimal objective value, denoted by $\overline{V}$. For the string case with nonhomogeneous sensors, this upper bound is given by
$$\overline{V} = H \max_{\s \in \Omega^{F}} f(\s) = Hf(\s_{(0)}).$$
For the set case with homogeneous sensors, this upper bound takes the form
$$\overline{V} = \sum_{k=0}^{H-1} f(\s_{(k)}),$$
where $\s_{(k)} \in \Omega^{F}$ is the placement location that will produce the $(k+1)^{\text{th}}$ largest $f(\s)$ value. We use $\beta_{2} = f(G_{H-1}) / \overline{V}$ to denote the performance bound obtained under our proposed framework. A comparison of these performance bounds is shown in Fig.~\ref{sensors_bound_adp}.

\begin{figure}[hbt!]
    \centering
    \includegraphics[width=0.98\columnwidth]{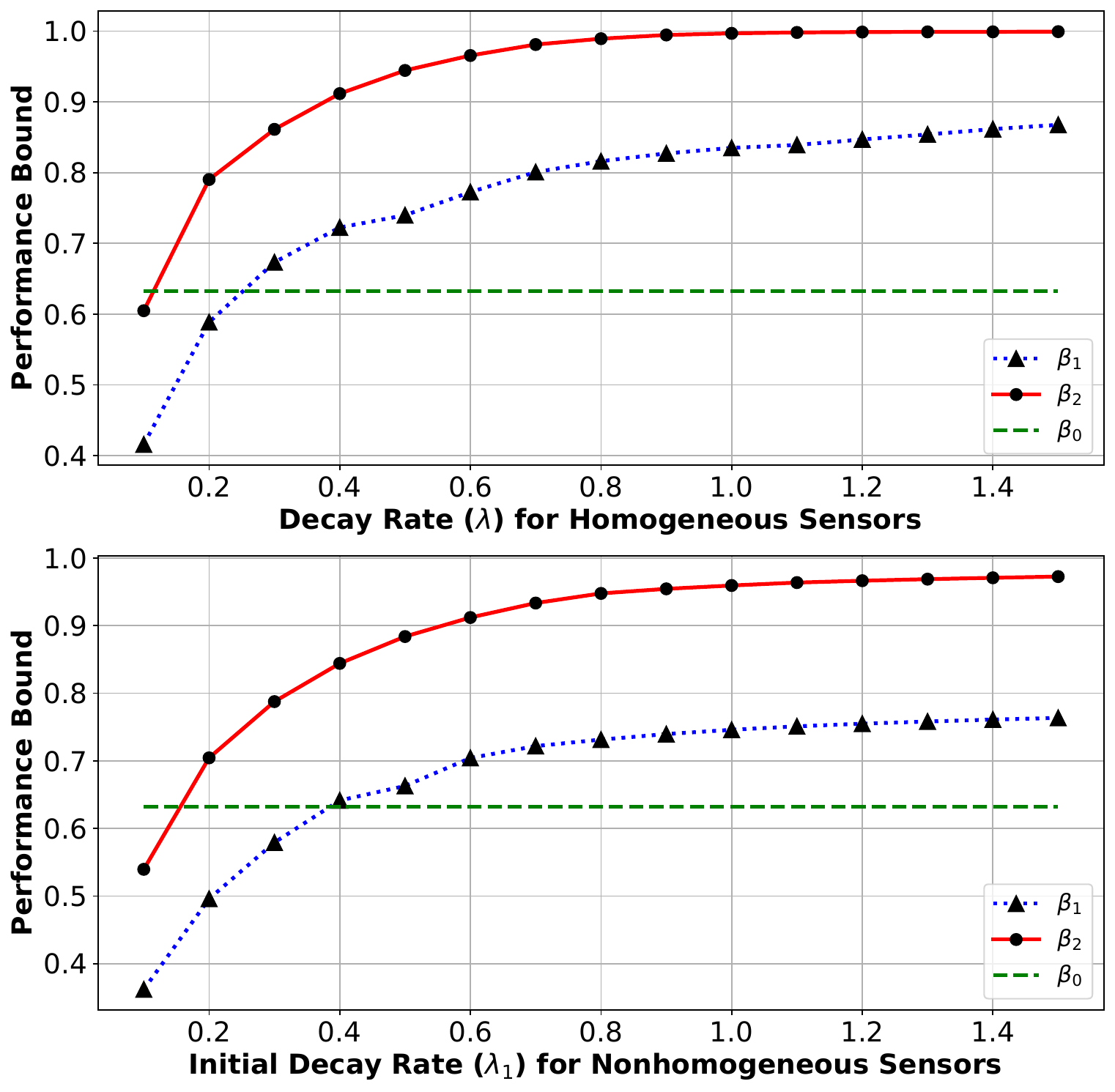}
    \caption{Performance bound comparison under different (initial) decay rates with number of placed sensors $H = 5$. \\
    Upper Figure: Homogeneous Sensors (Set Case); \\
    Lower Figure: Nonhomogeneous Sensors (String Case).}
    \label{sensors_bound_adp}
\end{figure}

We can observe from Fig.~\ref{sensors_bound_adp} that the red line ($\beta_{2}$) consistently dominates the blue line ($\beta_{1}$), which aligns with the theoretical result developed in \cite{van2025performance}. It also outperforms the green line ($\beta_{0}$) in many instances. Although the $\beta_{2}$ bound attains the same values as the top-$H$ bound we derived in \cite{van2025performance}, it is computed through our new bounding framework. This implies that our bounding framework generalizes the results in \cite{van2025performance}.

%%%%%%%%%%%%%%%%%%%%%%%%%%%%%%%%%%%%%%%%%%%%%%%%%%%%%%%%%%%%%%
%%%%%%%%%%%%%%%%%%%%%%%%%%%%%%%%%%%%%%%%%%%%%%%%%%%%%%%%%%%%%%
\section{Conclusion and Future Work}
\label{sec:conclusion}
We present the first computable ratio bound for approximate dynamic programming (ADP) schemes in finite-horizon sequential decision-making problems. We further show that how our bounding framework generalizes some results in the field of submodular optimization. Through applications to data-driven robot path planning, pendulum stabilization, and submodular multi-agent sensor coverage problems, we demonstrate the broad applicability of our proposed framework. Future work may include extending the bounding framework to multi-agent sequential decision-making settings, particularly those modeled by decentralized Markov decision processes.

%%%%%%%%%%%%%%%%%%%%%%%%%%%%%%%%%%%%%%%%%%%%%%%%%%%%%%%%%%%%%%
%%%%%%%%%%%%%%%%%%%%%%%%%%%%%%%%%%%%%%%%%%%%%%%%%%%%%%%%%%%%%%
% \appendices

% Appendixes, if needed, appear before the acknowledgment.

% \section*{Acknowledgment}

\section*{References}
\bibliographystyle{ieeetr}
\bibliography{bib}

\end{document}